\documentclass{article}
\usepackage{fullpage}
\usepackage{caption}
\usepackage{amsthm}
\theoremstyle{plain}
\newtheorem{theorem}{Theorem}
\newtheorem{lemma}[theorem]{Lemma}
\newtheorem{corollary}[theorem]{Corollary}
\theoremstyle{definition}
\newtheorem{definition}[theorem]{Definition}
\newtheorem{example}[theorem]{Example}
\theoremstyle{remark}
\newtheorem*{remark}{Remark}
\usepackage{amsmath,amssymb,stmaryrd}

\usepackage{tikz}
\usetikzlibrary{arrows,fit,decorations.pathreplacing,decorations.pathmorphing,positioning,patterns,shapes}

\usepackage{complexity}
\usepackage{xspace}
\usepackage{macros}
\usepackage{hyperref}

\title{Optimal Assumptions for Synthesis}
\author{Romain Brenguier (University of Oxford)}

\def\IA{\textsf{IA}}

\def\F{\textsf{F}}
\def\allpaths{\ensuremath{(\Sigma_I \cdot \Sigma_O)^\omega}}
\def\allhistories{\ensuremath{(\Sigma_I \cdot \Sigma_O)^*}}
\def\spec{\ensuremath{L}}
\def\adam{the environment\xspace}
\newcommand{\switch}[2]{\ensuremath{\left[#1\leftarrow#2\right]}}
\newcommand{\mypara}[1]{\medskip \noindent \textbf{#1}\ }

\begin{document}

\maketitle

\begin{abstract}
  Controller synthesis is the process of constructing a correct system automatically from its specification.
  This often requires assumptions about the behaviour of the environment.
  It is difficult for the designer to identify the assumptions that ensures the existence of a correct controller, and doing so manually can lead to assumptions that are stronger than necessary.
  As a consequence the generated controllers are sub optimal in terms of generality and robustness.
  In this work, given a specification, we identify the weakest assumptions that ensures the existence of a controller.
  We also consider two important classes of assumptions: the ones that can be ensured by the environment and assumptions that speaks only about inputs of the systems.
  We show that optimal assumptions correspond to strongly winning strategies, admissible strategies and remorse-free strategies respectively.
  Based on this correspondence, we propose an algorithm for computing optimal assumptions that can be ensured by the environment.
\end{abstract}

\section{Introduction}

Reactive systems are in outgoing interaction with their environment.
The goal of synthesis is the implementation of a correct reactive system from a high level specification.
A specification is given by a language over the alphabet of input and output signals of the desired system or program.
We will consider $\omega$-regular languages which are a powerful and natural way to describe the behaviour of reactive systems.
A specification is realisable if for all finite sequence of input signal we can produce at the same rate a sequence of output signals, such that the resulting sequence will belong to the language.
When this is the case we can implement a correct program with respect to the specification.
For regular languages this can be done using finite memory and thus implemented using Moore machines.

In general, the realisation of a specification requires some assumption about the behaviour of the environment.
In this work, given a specification, we compute the weakest assumption that makes it realisable.
Apart from looking for any assumption we also consider two important classes of assumptions: ensurable and input assumptions.
Ensurable assumptions are the one that can be ensured by the environment; in other term they cannot be falsified by a strategy of the controller.
These assumptions are natural to consider in a reactive environment context.
Input assumptions are independent of the sequence of output that is produced.
They are better suited than more general classes when the behaviour of the environment does not depend on the outputs of our system.

Synthesis is in general achieved by computation of winning strategies in a game.
For instance, if the specification is given by a parity automaton, we can see it has a game where the controller chooses output symbols and the adversary controls input symbols.
The existence of a winning strategy for the controller means that there is an implementation of the system such that all executions satisfies the specification and answers the realizability question.
When winning strategies do not exist, different classes have been introduced to characterise the ones that make a best-effort.
In particular, strongly winning strategy~\cite{Faella09} play a winning strategy as soon as the current history (sequence of signals seen so far) renders the existence of a winning strategy possible.
Admissible strategies~\cite{berwanger07} are not \emph{dominated} by other ones, in the sense that no strategy performs better than them against all adversary strategies.
Remorse-free strategies~\cite{DF11} are strategies for which no other performs better than them against all \emph{words} played by the adversary.
We draw a link between classes of assumptions and these classes of strategies.

\mypara{Example}
As an example, assume we want to design a sender on a network where packets can be lost or corrupted, and our goal is to obtain a protocol similar to the classical bit alternation protocol.
The outputs of the sender are actions $send_0$ and $send_1$, and the environment controls $ack_0$, $ack_1$ corresponding to acknowledgement of good reception of the packet.
The specification of the system are described by the $\omega$-regular expression:
$\Sigma_I \cdot (send_0 \cdot (\lnot ack_0 \cdot send_0)^* \cdot ack_0 \cdot send_1 \cdot (\lnot ack_1 \cdot send_1)^* \cdot ack_1)^\omega$.
Intuitively, we have to send message with bit control 0 until receiving the corresponding acknowledgement then do the same thing with the next message with bit control 1 and repeat this forever.

\noindent\begin{minipage}{0.55\textwidth}
  {\centering{
    \begin{tikzpicture}[xscale=1.8,yscale=1.4]
      \draw (0,0) node[draw,minimum size=6mm] (S1) {$s_1$};
      \draw (1,0) node[draw,circle,double] (S2){$s_2$};
      \draw (2,0) node[draw,minimum size=6mm] (S3) {$s_3$};
      \draw (2,-1) node[draw,circle] (S4) {$s_4$};
      \draw (3,0) node[draw,circle] (S5) {$s_5$};
      \draw (4,0) node[draw,minimum size=6mm] (S6) {$s_6$};
      \draw (4,-1) node[draw,circle] (S7) {$s_7$};
      \draw (2,-2) node[draw,minimum size=6mm] (BOT) {$\bot$};
      \draw (2,-3) node[draw,circle] (BOTO) {$\bot$};

      \draw [-latex'] (-0.5,0) -- (S1);
      \draw [-latex'] (S1) -- node[above] {$\Sigma_I$} (S2);
      \draw [-latex'] (S2) -- node[above] {$send_0$} (S3);
      \draw [-latex'] (S2) edge[bend right] node[left] {$\lnot send_0$} (BOT);
      \draw [-latex'] (S3) -- node[above] {$ack_0$} (S5);
      \draw [-latex'] (S3) edge[bend left] node[right] {$\lnot ack_0$} (S4);
      \draw [-latex'] (S4) edge[bend left] node[left] {$send_0$} (S3);
      \draw [-latex'] (S5) -- node[above] {$send_1$} (S6);
      \draw [-latex'] (S4) -- node[right] {$\lnot send_0$} (BOT);
      \draw  (S6.90) edge[bend right,-latex'] node[above] {$ack_1$} (S2.90);
      \draw [-latex'] (S6) edge[bend left] node[right] {$\lnot ack_1$} (S7);
      \draw [-latex'] (S7) edge[bend left] node[left] {$send_1$} (S6);
      \draw [-latex'] (S7) edge[bend left] node[right] {$\lnot send_1$} (BOT);
      \draw [-latex'] (BOT) edge[bend left] node[right] {$\Sigma_I$} (BOTO);
      \draw [-latex'] (BOTO) edge[bend left] node[left] {$\Sigma_O$} (BOT);
    \end{tikzpicture}
  }}
  \captionof{figure}{A specification by a B\"uchi automaton of the language~$\Sigma_I \cdot (send_0 \cdot (\lnot ack_0 \cdot send_0)^* \cdot ack_0 \cdot send_1 \cdot (\lnot ack_1 \cdot send_1)^* \cdot ack_1)^\omega$.
    Accepting states (with parity colour 0) are double lined.
    Square states mean that the next signal is an input, while circles mean it will be an output.}
  
  \label{fig:bit-alternation}
\end{minipage}
\hfill
\begin{minipage}{0.4\textwidth}
Although the implementation of the protocol seems straightforward, classical realizability fails here since if all packets are lost after some point the specification will not be satisfied.
To ensure realizability we have make the assumption that a packet that is repeatedly sent will eventually be acknowledged.
An admissible strategy for this specification can be implemented by a Moore machine which has the same structure as the automaton in \figurename~\ref{fig:bit-alternation} with output function $G$ such that $G(s_2) = G(s_4) = send_0$ and $G(s_5) = G(s_7) = send_1$.
This implementation is natural for the given specification and corresponds indeed to the bit alternation protocol.
The assumption corresponding to this strategy is the language recognised by the same automaton where we add $\bot$-states to the set of accepting states.
As we will see in Thm.~\ref{thm:non-dominated-nonrestrictive}, it is an optimal ensurable assumption for the specification.
\end{minipage}

\mypara{Scenarios} 
Giving specifications for our system is a way to disallow some behaviours that are not desirable.
Dually, we may want to specify execution scenarios that should be possible in the synthesised system.
We ask then for a system whose outcomes are all within the specifications and contains all the given scenarios.

\mypara{Generalisation}
Sometimes, some particular assumptions are natural for the problem we consider but we want to synthesise a system which is as robust as possible by generalising this assumption.
For instance, for the bit alternation protocol we could suggest as an initial assumption that two successive packets cannot be lost.
The synthesised systems would offer no guarantee two packets in a row are lost.
By generalising the assumption, we ensure that the strategy synthesised works well for the initial assumptions we have in mind, and for as many input sequences as possible.
For the bit alternation protocol the protocol works well if not all packet are lost after some point in time.


\mypara{Contribution}
In this article we establish correspondences between class of assumptions and classical classes of strategies.
These correspondences are summarised in the following table.
\begin{center}
\begin{tabular}{c c l}
  Class of assumption: &  Optimal achieved by: &\\
  \hline 
  General ($\mathcal{A}$) & strongly winning strategies & Thm.~\ref{thm:strongly-winning-optimal} \\
  Ensurable ($\mathcal{E}$) & admissible strategies & Thm.~\ref{thm:non-dominated-nonrestrictive}  \\
  Input ($\mathcal{I}$)  & remorsefree strategies & Thm.~\ref{thm:remorse-free-input} 
\end{tabular}
\end{center}
Based on these results, we show existence of optimal assumptions in most case and give algorithms to compute optimal assumptions.
In particular, we show:
\begin{itemize}
\item Existence of sufficient input assumptions compatible with a scenario is always true~(Thm.~\ref{thm:existence-optimal-compatible}).
  It is also true for safety assumption if the scenario is itself a safety language~(Thm.~\ref{thm:existence-optimal-compatible-safety}).
\item There may exist an infinite number of optimal and ensurable-optimal assumptions~(Thm.~\ref{thm:infinity-optimal}) and of input-optimal assumptions~(Thm.~\ref{thm:infinity}).
\item We can compute an optimal ensurable assumption in exponential time for parity specification and in polynomial time if we have an oracle to solve parity games~(Thm.~\ref{thm:Moore-ensurable-optimal} and Thm.~\ref{thm:compute-optimal}).
\item There is an exponential algorithm that given a sufficient assumption~$H$, generalises it by finding $H'$ such that $H \subseteq H'$ and $H'$ is ensurable-optimal~(Thm.~\ref{thm:generalization}).
\end{itemize}

\mypara{Comparison with previous works on assumptions for synthesis}
In \cite{CHJ08}, the study is focused on safety conditions defined by forbidding edges of the automaton defining specification $\spec$.
This approach is less general than ours since it depends on the choice of the automaton representing $\spec$.

\smallskip

\noindent\begin{minipage}{0.5\textwidth}
In the setting of \cite{CHJ08}, comparison between assumptions is based on the number of edges, while we compare them based on language inclusion which we find more relevant.
Consider the example of \figurename~\ref{fig:exCHJ08} taken from~\cite{CHJ08}.
Player 1 has no winning strategy from $s_1$.
According to~\cite{CHJ08}, there are 2 minimal sufficient assumptions which are $E'_s = \{(s_3,s_6)\}$ and $E'_s = \{(s_5,s_7)\}$.
However if we remove the edge from $s_3$ to $s_6$, $s_5$ is no longer accessible which means that the first assumption is in fact stronger than the second one, taking the point of view of language inclusion.
\end{minipage}
\hfill
\begin{minipage}{0.4\textwidth}
  \centering{
  \includegraphics[width=0.9\textwidth]{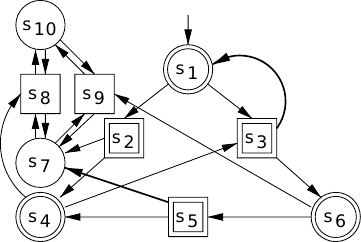}
  }
  \captionof{figure}{A game from~\cite{CHJ08}.}
  \label{fig:exCHJ08}
\end{minipage}

\bigskip

\noindent
\begin{minipage}{0.45\textwidth}
  \centering{
    \begin{tikzpicture}
      \draw (-1.3,0) node[draw,minimum size=6mm] (S1) {$s_1$};
      \draw (0,0) node[draw,circle] (S2) {$s_2$};
      \draw (1,1) node[draw,minimum size=6mm] (S3){$s_3$};
      \draw (1,-1) node[draw,minimum size=6mm] (S4) {$s_4$};
      \draw (2.5,-1) node[draw,circle,double] (S5) {$s_5$};
      \draw (4,-1) node[draw,double,minimum size=6mm] (S8) {$s_8$};

      \draw (2.5,1) node[draw,circle] (S6) {$s_6$};
      \draw (4,1) node[draw,minimum size=6mm] (S9) {$s_7$};
      \draw [-latex'] (-2,0) -- (S1);
      \draw [-latex'] (S1) -- node[above] {$\Sigma_I$} (S2);
      \draw [-latex'] (S2) -- node[above] {$o_1$} (S3);
      \draw [-latex'] (S2) -- node[above] {$o_2$} (S4);
      \draw [-latex'] (S3) -- node[pos=0.3,above] {$i_1$} (S5);
      \draw [-latex'] (S3) -- node[above] {$i_2$} (S6);

      \draw [-latex'] (S4) -- node[above] {$i_2$} (S5);
      \draw [-latex'] (S4) -- node[pos=0.3,above] {$i_1$} (S6);
      \draw [-latex'] (S5) edge[bend left] node[above] {$\Sigma_O$} (S8);
      \draw [-latex'] (S8) edge[bend left] node[below] {$\Sigma_I$} (S5);
      \draw [-latex'] (S6) edge[bend left] node[above] {$\Sigma_O$} (S9);
      \draw [-latex'] (S9) edge[bend left] node[below] {$\Sigma_I$} (S6);
    \end{tikzpicture}
  }
  
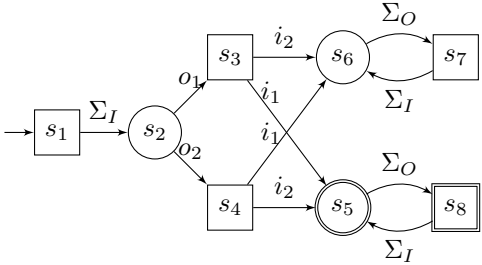
\captionof{figure}{Automaton for specification $\Sigma_I \cdot (o_1 \cdot i_1 \mid o_2 \cdot i_2) \cdot \Sigma_O \cdot \allpaths$.}
  \label{fig:counter-example}
\end{minipage}
\hfill
\begin{minipage}{0.5\textwidth}
Moreover, we show that even for this restricted class of safety assumption, the claim that there is a unique optimal assumption (\cite[Thm.~5]{CHJ08}) is wrong.
Consider the example of \figurename~\ref{fig:counter-example}.
Removing $(s_3,s_6)$ or $(s_4,s_6)$ is sufficient for $\spec = \Sigma_I \cdot (o_1 \cdot i_1 \mid o_2 \cdot i_2) \cdot \Sigma_O \cdot \allpaths$.
Note that both these assumptions are ensurable since they lead to immediate violation of the safety objective.
This therefore contradicts~\cite[{Thm.~5}]{CHJ08} which claims that if $\spec \ne \varnothing$ then there exists a unique minimal safety assumption that is ensurable and sufficient for $\spec$.
\end{minipage}

\section{Preliminaries}

Given a finite alphabet $\Sigma$ and an infinite word $w \in \Sigma^\omega$, we use $w_i$ to denote the $i$-th symbol of $w$,
and $w_{\le i} = w_1 \cdots w_i$ the finite prefix of $w$ of length $i$.
We write $|w_{\le i}| = i$ its length.

\subsection{Classical realizability}

A reactive system reads \emph{input signals} in a finite alphabet $\Sigma_I$ and produces \emph{output signals} in a finite alphabet $\Sigma_O$.
We assume for the rest of this paper that these alphabets are fixed.
A \emph{specification} of a reactive system is an $\omega$-regular language $L \subseteq (\Sigma_I \cdot \Sigma_O)^\omega$.
A \emph{program} or \emph{strategy} is a mapping $\sigma_\exists \colon (\Sigma_I \cdot \Sigma_O)^* \cdot \Sigma_I \mapsto \Sigma_O$.
An \emph{outcome} of such a strategy $\sigma_\exists$ is a word $w$ such that for all $i\in \mathbb{N}$, $w_{2\cdot i+2} = \sigma_\exists (w_{\le 2 \cdot i+1})$.
We write $\Out(\sigma_\exists)$ for the set of outcomes of $\sigma_\exists$.

\mypara{Realizability problem~\cite{PR89}}
Given a specification~$\spec$,
the realizability problem asks whether there exists
a strategy $\sigma_\exists$ such that
$\Out(\sigma_\exists) \subseteq \spec$.
Such a strategy is said \emph{winning} for $\spec$.
The process of constructing such a strategy is called \emph{synthesis}.

\mypara{Parity automata}
We will assume that specifications are given by deterministic parity automata, which can recognise any $\omega$-regular languages~\cite{GTW02}.
A \emph{parity automaton} is given by $\langle S, s_0, \Delta, \chi \rangle$, where $S$ is a finite set of states, $s_0 \in S$ is the initial state, $\Delta \in S \times (\Sigma_I \cup \Sigma_O) \times S$ is the transition relation, and $\chi \colon S \to \mathbb{N}$ is a colouring function.
A path~$\rho\in S^\omega$ is accepting if the smallest colour seen infinitely often is even (i.e. if $\min\{ c \mid \forall i\in \mathbb{N}.\ \exists j\ge i.\ \chi(w_j) = c \} \in 2 \cdot \mathbb{N}$).
A word~$w$ is accepted if there is an accepting path whose labelling is $w$.
A \emph{B\"uchi automaton} is a parity automaton for which $\chi(S) \subseteq \{ 0, 1 \}$.
A \emph{safety automaton} is a B\"uchi automaton where states of colour $1$ are absorbing.
The language recognised by an automaton is the set of words it accepts.
In practice specification are often given in temporal logics such as LTL before being translated to an automaton representation.
In our examples, we will sometimes mention LTL formulas, using the syntax $\X \phi$ meaning $\phi$ holds in the next state, $\phi_1 \U \phi_2$ meaning $\phi_1$ holds until $\phi_2$ holds (and $\phi_2$ must hold at some point), $\F \phi := \true\ \U\ \phi$ and $\G \phi := \lnot \F (\lnot \phi)$.

\mypara{Strategies}
The realizability problem is best seen as a game between two players~\cite{FJR09}. 
The environment chooses the input signals and the controller the output signals.
We therefore also define the concept of \emph{environment-strategy} which is a mapping $\sigma_\forall \colon (\Sigma_I \cdot \Sigma_O)^* \mapsto \Sigma_I$.
Given an environment-strategy $\sigma_\forall$, we write $\Out(\sigma_\forall)$ the word $w$ such that for all $i\in \mathbb{N}$, $w_{2\cdot i+1} = \sigma_\forall(w_{\le 2 \cdot i})$.
Given an input word $u \in {\Sigma_I}^\omega$, we write $\Out(\sigma_\exists,u)$ the unique outcome such that for all $i\in \mathbb{N}$, $w_{2\cdot i+1} = u_{i+1}$ and $w_{2\cdot i+2} = \sigma_\exists(w_{\le 2\cdot i+1})$.
We also write $\Out(\sigma_\exists,\sigma_\forall) = \Out(\sigma_\exists) \cap \Out(\sigma_\forall)$, note that it contains only one outcome.
A finite prefix of an outcome is called a \emph{history}.
Given a history~$h$, we write $\Out_h(\sigma_\exists)$ a word $w$ such that for all $i\le |h|$, $w_i = h_i$ and for all $i$ such that $2\cdot i+2> |h|$, $w_{2\cdot i+2} = \sigma_\exists (w_{\le 2 \cdot i+1})$.

We write $\pi_I$ and $\pi_O$ the \emph{samplings} over input and output signals respectively, that is $\pi_I \colon (\Sigma_I \cdot \Sigma_O)^\omega \to {\Sigma_I}^\omega$ is such that $\pi_I(w)_i = w_{2 \cdot i-1}$ and $\pi_O \colon (\Sigma_I \cdot \Sigma_O)^\omega \to {\Sigma_O}^\omega$ is such that $\pi_O(w)_i = w_{2 \cdot i}$.

\mypara{Moore machine}
In practice strategies are implemented using Moore machines, that correspond to strategies that only use finite memory.
  A \emph{Moore machine} is given by $\langle S_I,S_O, s_0, \delta, G \rangle$ where $S=S_I \cup S_O$ is a finite set of states, $S_I$ is a set of input states and $S_O$ of output states, $s_0 \in S_I$ is the initial state, $\delta \in S \times \Sigma_I \cup \Sigma_O \to S$ is the transition function, and $G \colon S_O \to \Sigma_O$ is an output function.
  A Moore machine implements a strategy~$\sigma_\exists$ where for all history $h\in \allhistories \cdot \Sigma_I$, $\sigma_\exists(h) = G(s_{|h|})$ where for all $0 \le i < |h|$, $s_{i+1} = \delta(s_i,h_{i+1})$.
  Note that our Moore machine read both inputs and outputs.
  In many application there would be no need to read the outputs since it can be left undefined on histories that are incompatible with the strategy.
  However we prefer this definition which is coherent with our definition of strategies and makes it easier to combine strategies which may not be compatible with the same histories.


\medskip

\noindent\begin{minipage}{0.52\textwidth}
\begin{example}
  In all the examples of this article we will assume that the set of input signals is $\Sigma_I = \{ i_1 , i_2 \}$ and the set output of output signals is $\Sigma_O = \{ o_1 , o_2 \}$.
  The specification given in \figurename~\ref{fig:simple-example} is realisable.
  The corresponding winning strategy consists for the controller to invert the input signals: if $i_1$ is the input then in next round we do not output $o_1$ and vice-versa.
  Formally, this strategy is given for all history $h \in (\Sigma_I \cdot \Sigma_O) \cdot \Sigma_I$ by if $h_{|h|} = i_1$ then $\sigma_\exists(h) = o_2$ and if $h_{|h|} = i_2$ then $\sigma_\exists(h) = o_1$.
  A Moore machine which implement a winning strategy can also be obtained from the parity automaton by setting the output function to be such that~$G(s_2)=o_2$, $G(s_1)=o_1$.
\end{example}
\end{minipage}\hfill
\begin{minipage}{0.45\textwidth}
  \centering{
    \begin{tikzpicture}[xscale=1.5,yscale=1.5]
      \draw (0,0) node[draw,minimum size=6mm,double] (S1) {$s_1$};
      \draw (1,1) node[draw,circle,double] (S2){$s_2$};
      \draw (1,-1) node[draw,circle,double] (S3) {$s_3$};
      \draw (2,0) node[draw,minimum size=6mm] (S4) {$s_4$};
      \draw (3,0) node[draw,circle] (S5) {$s_5$};

      \draw [-latex'] (-0.7,0) -- (S1);
      \draw [-latex'] (S1) -- node[below] {$i_1$} (S2);
      \draw [-latex'] (S1) -- node[below] {$i_2$} (S3);
      \draw [-latex'] (S2) edge[bend right] node[above] {$o_2$} (S1);
      \draw [-latex'] (S2) -- node[above] {$o_1$} (S4);
      \draw [-latex'] (S3) edge[bend right] node[above] {$o_1$} (S1);
      \draw [-latex'] (S3) -- node[above] {$o_2$} (S4);
      \draw [-latex'] (S4) edge[bend right] node[below] {$\Sigma_I$} (S5);
      \draw [-latex'] (S5) edge[bend right] node[above] {$\Sigma_O$} (S4);
    \end{tikzpicture}
  }
  
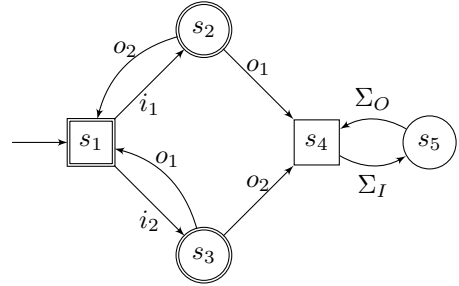
\captionof{figure}{A specification by a B\"uchi automaton corresponding to LTL formula~$\G(i_1 \Leftrightarrow \X (\lnot o_1))$.
    Accepting states (with parity colour 0) are double lined.
    Square states mean that the next signal is an input, while circles mean it will be an output.}
  \label{fig:simple-example}
\end{minipage}

\subsection{Assumptions}
An \emph{assumption} is an $\omega$-regular language $A\subseteq (\Sigma_I \cdot \Sigma_O)^\omega$.
Assumption $A$ is \emph{sufficient} for specification $\spec$ if there is a strategy~$\sigma_\exists$ of the controller such that any outcome either satisfies $\spec$ or is not in $A$, i.e. $\Out(\sigma_\exists) \cap A \subseteq \spec$.
In that case we also say that $A$ is \emph{sufficient for $\sigma_\exists$}.
We are looking for assumptions that are the least restrictive.
We say that assumption $A$ is \emph{less restrictive} than $B$ if 
$A \subseteq B$.
We say it is \emph{strictly less restrictive} if $A \subset B$ (i.e. $A\subseteq B$ and $A \ne B$).

We will consider several class of assumptions:
\begin{itemize}
\item The general class of assumptions is written $\mathcal{A} = (\Sigma_I \cdot \Sigma_O)^\omega$.
\item An \emph{input-assumption} is an assumption which concerns only inputs and does not restrict at all outputs.
  $A$ is an input-assumption if for all words $w$ and $w'$, $\pi_I(w) = \pi_I(w') \land w \in A \Rightarrow w' \in A$.
  We write this class $\mathcal{I} = \{ A \in \mathcal{A} \mid \forall w,w'\in \allpaths.\ \pi_I(w) = \pi_I(w') \land w \in A \Rightarrow w' \in A \}$.
\item 
  We say that $A$ is \emph{ensurable}, if the environment has a winning strategy for $A$, i.e.
  for each $w \in \allhistories$, if $w \cdot \allpaths \cap A \ne \varnothing$ then
  there exists an environment strategy $\sigma_\forall$ such that $w \cdot \allpaths \cap \Out(\sigma_\forall) \ne \varnothing$ and $\Out(\sigma_\forall) \subseteq A$.
  We write $\mathcal{E}$ for this class.
  The fact that \adam can ensure the assumption is often a requirement done in synthesis~(see for instance~\cite{BEJK14}).
\item We say that an assumption~$A$ is \emph{output-restrictive} if it restricts the strategies of the controller, that is
  if there is $w \in (\Sigma_I \cdot \Sigma_O)^*$ and $\sigma_\exists$ a strategy, such that 
\[
\left\{
\begin{array}{l r l}
  & A \cap w \cdot \allpaths  & \ne \varnothing \\
  \text{and} & \Out(\sigma_\exists) \cap w \cdot \allpaths & \ne \varnothing \\
  \text{and} & A \cap \Out(\sigma_\exists) \cap w \cdot \allpaths  & = \varnothing 
\end{array}
\right.
\]
Intuitively an output-restrictive assumption $A$ forbids strategy $\sigma_\exists$, so playing $\sigma_\exists$ would be a trivial way to satisfy $A \Rightarrow \spec$.
From the point of view of synthesis it is better to avoid such assumptions.
We write this class $\mathcal{O}$.
\item Given an assumption $A$, we write the set of bad prefixes $\Bad(A) = \{ h\in \allhistories \mid h\cdot \allpaths \cap A = \varnothing \}$.
  Assumption $A$ is a \emph{safety assumption} if every word not in $A$ has a bad prefix~\cite{FK15}, i.e. $A = \allpaths \setminus \{ w \mid \exists k.\ w_{\le 2 \cdot k} \in \Bad(A)\}$.
  We write $\mathcal{S}$ for this class. 
\end{itemize}

For a class~$\mathcal{C}$ of assumptions, we say that assumption $A$ is \emph{$\mathcal{C}$-optimal} for $\spec$ if $A$ belongs to $\mathcal{C}$, is sufficient for $\spec$ and there is no assumption $B \in \mathcal{C}$ that is strictly less restrictive than $A$ and sufficient for $\spec$.

\begin{remark}
  Note that $\spec$ is always a sufficient assumptions, however it is too strong and will never be interesting for synthesis: if we assume that our specification always hold then there is nothing left to do.
  That is why we ask for assumptions that are as weak as possible.
\end{remark}

We now prove the relationships that exist between the classes of assumptions.
\begin{lemma}
  Non-empty input assumptions are ensurable, i.e. $(\mathcal{I}\setminus \{\varnothing\}) \subset \mathcal{E}$
\end{lemma}
\begin{proof}
  $(\mathcal{I}\setminus \{\varnothing\}) \subseteq \mathcal{E}$:
  Let $A$ be a non-empty input-assumption, and $w \in A$.
  Let $\sigma_\forall$ be the strategy of the environment that follows $w$, that is for all histories~$h$, $\sigma_\forall(h) = w_{|h|+1}$.
  For all outcome $u \in \Out(\sigma_\forall)$, $\pi_I(u) = \pi_I(w)$.
  Since $\pi_I(w) = \pi_I(u)$ and $A$ is an input assumption $u \in A$.
  Hence $A$ is ensurable.

  $(\mathcal{I}\setminus \{\varnothing\}) \ne \mathcal{E}$:
  Consider the specification $\spec = \X o_1 \Leftrightarrow \X\X i_1$, it is ensurable by any strategy such that $\sigma_\forall(\Sigma_I \cdot o_1) = i_1$ and $\sigma_\forall(\Sigma_I \cdot o_2) = i_2$.
  However it is not an input-assumption: $w = (i_1 \cdot o_1)^\omega \in \spec$ and $w' = (i_1 \cdot o_2)^\omega \not\in \spec$ but $\pi_I(w) = \pi_I(w')$.
\end{proof}

\begin{lemma}\label{lem:ensurable-nonrestrictive}
  For $\omega$-regular specifications, ensurable assumptions are exactly the assumptions that are not-output-restrictive, $\mathcal{E} = \mathcal{G} \setminus \mathcal{O}$.
\end{lemma}
\begin{proof}
  $\mathcal{E} \subseteq \mathcal{G} \setminus \mathcal{O}$:
  Let $A$ be an ensurable assumption, and $w \in \allhistories$.
  Assume there is a strategy $\sigma_\exists$ such that $A \cap w \cdot \allpaths \ne \varnothing$ and $\Out(\sigma_\exists) \cap  w \cdot \allpaths \ne \varnothing$, we prove that $A \cap \Out(\sigma_\exists) \cap  w \cdot \allpaths \ne \varnothing$, which shows that $A$ is nonrestrictive.
  Let $u \in A \cap w \cdot \allpaths$ by the definition of ensurable assumptions there exists $\sigma_\forall$ compatible with $w$ and such that $\Out(\sigma_\forall) \cap A$.
  Let $w' = \Out(\sigma_\exists,\sigma_\forall)$, $w$ is a prefix of $w'$ since it is compatible with both $\sigma_\exists$ and $\sigma_\forall$.
  Since $\Out(\sigma_\forall) \cap A$, $w' \in A$ so $w \in A \cap \Out(\sigma_\exists) \cap w \cdot \allpaths$.
  Hence $A$ is nonrestrictive.

  $\mathcal{G} \setminus \mathcal{O} \subseteq \mathcal{E}$:
  Let $A$ be an assumption which is not output restrictive.
  Assume towards a contradiction that there is some history $w$ such that $w \cdot \allpaths \cap A \ne \varnothing$ and environment has no winning strategy~$\sigma_\forall$ for $A$ which is compatible with $w$.
  This means environment has no winning strategy for $w^{-1} \cdot A = \{ w' \in \allpaths \mid w \cdot w' \in A \}$.
  By determinacy of $\omega$-regular objectives~\cite[Corollary~2.10]{GTW02} there is a winning strategy $\sigma_\exists$ for $\allpaths \setminus (w^{-1} \cdot A)$.
  Consider the strategy $\sigma'_\exists$ that plays according to $w$ and then switches to $\sigma_\exists$, that is for a history $h$ if $|h| < |w|$ then $\sigma'_\exists(h) = w_{|h|+1}$ and if $|h| \ge |w|$, $\sigma'_\exists(h) = \sigma_\exists(h_{\ge |w|}$.
  The history $w$ is compatible with $\sigma'_\exists$ so $\Out(\sigma'_\exists) \cap w \cdot \allpaths \ne \varnothing$.
  But since $\sigma_\exists$ is winning for $\allpaths \setminus (w^{-1}\cdot A)$, no outcome is of the form $w \cdot w'$ with $w' \in w^{-1} \cdot A$.
  This means that $A \cap \Out(\sigma'_\exists) \cap w \cdot \allpaths= \varnothing$ which contradicts the fact that $A$ is not output restrictive.
\end{proof}
\begin{example}
  The specification $\spec = o_1 \U i_1$ is not realisable, however several assumptions can be sufficient for it.
  The automaton corresponding to $\spec$ is represented in \figurename~\ref{fig:o1Ui1}.
  Consider for instance the assumption~$A$ given by the LTL formula $\F o_1$.
  It is sufficient for $\spec$ and is in fact sufficient for any specification since a strategy~$\sigma_\exists$ of \eve which never plays $o_1$ has no outcome in $A$.
  To avoid this degenerate assumptions we focus on non-restrictive assumptions: $\F o_1$ is indeed output restrictive.
  On the other hand $\F(o_1) \Leftrightarrow \F(i_1)$ is ensurable because \adam can react to make the assumption hold, no matter the strategy $\sigma_\exists$ we chose.
  We can also check that it is sufficient for $\spec$: the strategy that always play $o_1$ is winning.
  This is a fine assumption in the context of a reactive environment, but if the behaviour of the environment is independent of the output of the system, talking about $o_1$ in the assumption is not relevant.
  In that case we would prefer for instance the input-assumption $\F i_1$ which is sufficient for $\spec$ and is independent of outputs.
\end{example}
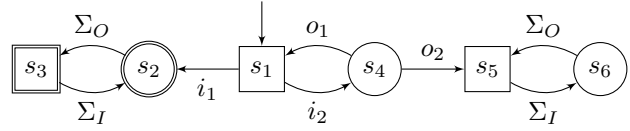
\begin{figure}[h]
  \centering{
    \begin{tikzpicture}[xscale=1.5,yscale=1.5]
      \draw (0,0) node[draw,minimum size=6mm] (S1) {$s_1$};
      \draw (-1,0) node[draw,circle,double] (S2){$s_2$};
      \draw (-2,0) node[draw,minimum size=6mm,double] (S3) {$s_3$};
      \draw (1,0) node[draw,circle] (S4) {$s_4$};
      \draw (2,0) node[draw,minimum size=6mm] (S5) {$s_5$};
      \draw (3,0) node[draw,circle] (S6) {$s_6$};

      \draw [-latex'] (0,0.6) -- (S1);
      \draw [-latex'] (S1) -- node[below] {$i_1$} (S2);
      \draw [-latex'] (S2) edge[bend right] node[above] {$\Sigma_O$} (S3);
      \draw [-latex'] (S3) edge[bend right] node[below] {$\Sigma_I$} (S2);
      \draw [-latex'] (S1) edge[bend right]  node[below] {$i_2$} (S4);
      \draw [-latex'] (S4) edge[bend right]  node[above] {$o_1$} (S1);
      \draw [-latex'] (S4) -- node[above] {$o_2$} (S5);
      \draw [-latex'] (S5) edge[bend right]  node[below] {$\Sigma_I$} (S6);
      \draw [-latex'] (S6) edge[bend right]  node[above] {$\Sigma_O$} (S5);
    \end{tikzpicture}
  }
  \caption{B\"uchi automaton recognising the language corresponding to LTL formula~$o_1 \U i_1$.}
  \label{fig:o1Ui1}
\end{figure}

\subsection{Refinement using scenarios}

As we will see in the next sections, in general there are an infinite number of incomparable optimal assumptions.
This brings the problem of choosing one among all the possibilities.
How can the algorithm know which assumption will satisfy the designer?
A solution is to get some feedback from the user in the form of \emph{scenarios}.
A scenario is a behaviour that the environment can exhibit in practice and that the strategy produced in the end should allow.
We will use scenarios provided by the user to chose a solution.

\mypara{Scenario}
Formally a \emph{scenario} is given by a language $S \subseteq \allpaths$.
A strategy $\sigma_\exists$ is \emph{compatible} with the set of scenarios $S$ when $S \subseteq \Out(\sigma_\exists)$.
Similarly $S$ is \emph{compatible} with specification $\spec$ if $S \subseteq \spec$.
A scenario $S$ is \emph{coherent} 
if there is no history $h \in \allhistories \cdot \Sigma_I$ prefixes of two words $w,w' \in S$ such that $w_{|h|+1} \ne w'_{|h|+1}$.
We say that assumption $H$ is sufficient for $\spec$ and $S$, if there exists $\sigma_\exists$ such that $S \subseteq H \cap \Out(\sigma_\exists) \subseteq \spec$.

If $S$ is not coherent then a compatible strategy in $h$ would need to allow both $w_{|h|+1}$ and $w'_{|h|+1}$ which is impossible.
Being coherent is in fact a necessary and sufficient condition for the existence of a compatible strategy.
\begin{lemma}\label{lem:scenario-compatible}
  Scenario $S$ is coherent if, and only if, there exists a strategy compatible with $S$.
  In particular, given a coherent scenario $S$ and a strategy $\sigma_\exists$, 
  the strategy $[S \rightarrow \sigma_\exists]$ defined by:
  \[[S \rightarrow \sigma_\exists] (h) = \left\{\begin{array}{ll}
  w_{h+1} & \textrm{if there is } w \in S \text{ such that } h \text{ is a prefix of } w\\
  \sigma_\exists(h) & \textrm{otherwise}
  \end{array}\right.\]
  is compatible with $S$.
\end{lemma}
\begin{proof}
  Let $S$ be a coherent scenario, we show that for any strategy $\sigma_\exists$, $[S \rightarrow \sigma_\exists]$ is compatible with $S$.
  Let $w \in S$, and $i \in \mathbb{N}$.
  The history $w_{2 \cdot i +1}$ is prefix of a scenario and therefore there is $w' \in S$ such that $[S\rightarrow \sigma_\exists](w_{2 \cdot i +1}) = w'_{2\cdot i +2}$.
  Since $S$ is coherent $w_{2\cdot i +2} = w'_{2\cdot i +2}$ and $[S\rightarrow \sigma_\exists]$ is compatible with $w_{\le 2\cdot i +2}$.
  This shows $w$ is compatible with $[S \rightarrow \sigma_\exists]$.

  Now we will show that if $S$ is compatible with some strategy $\sigma_\exists$ then $S$ is coherent.
  Let $h$ be a prefix of some words $w,w' \in S$.
  Since $S$ is compatible with $\sigma_\exists$, $w_{|h|+1} = \sigma_\exists(h) = w'_{|h|+1}$.
  Hence $S$ is coherent.
\end{proof}

\begin{example}
  Consider the example in \figurename~\ref{fig:infinity}.
  There are a lot of different possible assumptions we could chose from.
  However if we give the scenario $\Sigma_I \cdot o_2 \cdot i_2 \cdot \Sigma_O \cdot \allpaths$ then the only optimal assumption is 
 $(\Sigma_I \cdot o_2 \cdot i_2 + \Sigma_I \cdot o_1 \cdot \Sigma_I) \cdot \Sigma_O \cdot \allpaths$.
  The corresponding winning strategy consists in playing $o_2$ for the first output.
\end{example}

\begin{figure}[h]
  \centering{
    \begin{tikzpicture}
      \draw (-2,0) node[draw,minimum size] (S0) {$s_0$};
      \draw (0,0) node[draw,minimum size=6mm] (S1) {$s_1$};
      \draw (2,0) node[draw,circle] (S2){$s_2$};
      \draw (4,0.8) node[draw,double,minimum size=6mm] (S3) {$s_3$};
      \draw (4,-0.8) node[draw,minimum size=6mm] (S4) {$s_4$};
      \draw (6,0.8) node[draw,double,circle] (S5) {$s_5$};
      \draw (6,-0.8) node[draw,circle] (S6) {$s_6$};

      \draw [-latex'] (-3,0) -- (S0);
      \draw [-latex'] (S0) edge[bend left] node[above] {$i_1,i_2$} (S1);
      \draw [-latex'] (S1) edge[bend left] node[below] {$o_1$} (S0);
      \draw [-latex'] (S1) -- node[above] {$o_2$} (S2);
      \draw [-latex'] (S2) -- node[above] {$i_1$} (S4);
      \draw [-latex'] (S2) -- node[above] {$i_2$} (S3);
      \draw [-latex'] (S3) edge[bend left] node[above] {$o_1,o_2$} (S5);
      \draw [-latex'] (S5) edge[bend left] node[above] {$i_1,i_2$} (S3);
      \draw [-latex'] (S4) edge[bend left] node[above] {$o_1,o_2$} (S6);
      \draw [-latex'] (S6) edge[bend left] node[above] {$i_1,i_2$} (S4);
    \end{tikzpicture}
  }
  \caption{B\"uchi automaton for specification~$(\lnot o_2) \U (o_2 \land \X i_2)$. }
  \label{fig:infinity}
\end{figure}
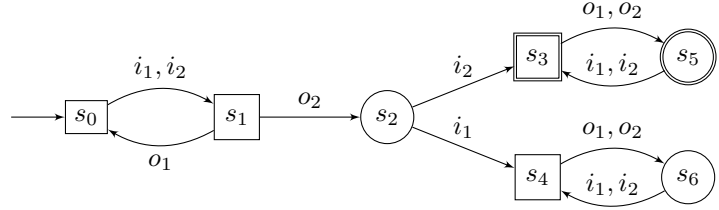

We say that an assumption~$A$ is $\mathcal{C}$-optimal for $\spec$ and $S$ if it is sufficient for $\spec$ and $S$ and there is no $A' \in \mathcal{C}$ strictly less restrictive than $A$ and sufficient for $\spec$ and $S$.

\subsection{Existence of a sufficient assumption with scenario}

We show that there always exist an input assumption which is sufficient and compatible with a given scenario.
\begin{theorem}\label{thm:existence-optimal-compatible}
  Let $\spec$ be a specification and $S$ a scenario.
  If $S$ is compatible with $\spec$, then there exists a 
  input assumption 
  which is sufficient for $\spec$ and compatible with $S$.
\end{theorem}
\begin{proof}
  Assume $S$ is compatible with $\spec$ and let $H = \{ w \mid \exists w' \in S.\ \pi_I(w) = \pi_I(w') \}$.
  It is clear that $H$ is an input assumption.
  We will prove that for any strategy~$\sigma_\exists$, $H$ is sufficient for $[S \rightarrow \sigma_\exists]$.
  Let $w \in H\cap \Out([S\rightarrow \sigma_\exists])$.
  Since $w \in H$, there is $w' \in S$ such that $\pi_I(w') = \pi_I(w)$.
  Since $[S\rightarrow \sigma_\exists]$ is compatible with $S$ (Lem.~\ref{lem:scenario-compatible}), it is compatible with $w'$, and therefore $w' = \Out(\sigma_\exists, \pi_I(w)) = w$.
  This proves that $w \in S$ and therefore $H$ is sufficient for $\spec$.
\end{proof}

\begin{theorem}\label{thm:existence-optimal-compatible-safety}
  Let $\spec$ be a specification and $S$ a scenario compatible with $\spec$. If $S$ is a safety language then there exists a safety assumption which is sufficient for $\spec$ and compatible with $S$.
\end{theorem}
\begin{proof}
  The set of bad prefixes of $S$ define a safety language $A = \allpaths \setminus \{ w \mid \exists k.\ w_{\le 2 \cdot k} \in \Bad(S) \}$.
  Let $\sigma_\exists$ be a strategy, and $w \in A \cap \Out([S \rightarrow \sigma_\exists])$.
  There is no $i$ such that $w_{\le i} \in \Bad(S)$.
  Therefore for all $i$, $w_{\le i}$ is the prefix of some $w' \in S$.
  Hence for all $i$, $w_{\le i}$ is not in $\Bad(S)$,
  and since $S$ is a safety language this means that $w \in S$ and since $S$ is compatible with $\spec$, $w \in \spec$.  
  Thus $A$ is sufficient for $\spec$.
\end{proof}

\begin{remark}
  There are example of scenarios that are not safety languages for which there is no sufficient safety assumption.
  Consider $S$ given by $\F i_1$ and $\spec = \F i_1$. 
  A safety assumption $H$ compatible with $S$ has to contain $\F i_1$.
  Assume towards a contradiction that there is $w\not\in H$.
  Since $H$ is a safety assumption, $w$ has a bad prefix, that is there is $i$ such that $w_{\le 2 \cdot i} \cdot \allpaths \cap H = \varnothing$.
  As $w_{\le 2 \cdot i} \cdot (i_1 \cdot o_1)^\omega \in S$, this is contradiction with the fact that $H$ is compatible with $S$.
  Therefore $H = \allpaths$ and this is not sufficient for $\spec$.
\end{remark}

\section{General assumptions}\label{sec:general}

In this section we study general assumptions, without distinguishing the ones that are restrictive.
Properties established in this section will be useful when studying ensurable assumptions.

\subsection{Necessary and sufficient assumptions for a strategy}
Given a specification $\spec$ and a strategy $\sigma_\exists$, we say that an assumption $A$ is \emph{necessary} for $\sigma_\exists$ if $B$ sufficient for $\sigma_\exists$ implies $B \subseteq A$.

\def\EA{\textsf{EA}}
\begin{lemma}\label{lem:ea-sufficient}
  Given a strategy $\sigma_\exists$ of \eve, the assumption 
  \(\EA(\sigma_\exists) = \spec \cup \left((\Sigma_I \cdot \Sigma_O)^\omega \setminus \Out(\sigma_\exists) \right)\)
  is sufficient and necessary for $\sigma_\exists$.
\end{lemma}
\begin{proof}
  Let us prove that $\EA(\sigma_\exists)$ is sufficient for $\sigma_\exists$.
  Let $w\in \Out(\sigma_\exists)$, either $w$ is winning for condition $\spec$, or it belongs to $\Out(\sigma_\exists) \setminus \spec$.
  In this second case, it does not satisfy assumption $\EA(\sigma_\exists)$ which is thus sufficient for $\sigma_\exists$.

  Let us now assume that some assumption $A$ is sufficient for $\sigma_\exists$ and prove that then $A\subseteq \EA(\sigma_\exists)$, which shows $\EA(\sigma_\exists)$ is necessary.
  Let $w \in A$.
  If $w \in \Out(\sigma_\exists)$ then since $A$ is sufficient for $\sigma_\exists$ we must have that $w \models \spec$ and therefore $w \in \EA(\sigma_\exists)$ by definition of $\EA$.
  Otherwise $w\not\in\Out(\sigma_\exists)$, then by definition of $\EA$, $w \in \EA(\sigma_\exists)$.
\end{proof}

\begin{corollary}\label{cor:optimal-ea}
  If $A$ is optimal then there exists $\sigma_\exists$ such that $A = \EA(\sigma_\exists)$.
\end{corollary}
\begin{proof} 
  Since $A$ is sufficient, there is $\sigma_\exists$ such that $A$ is sufficient for $\sigma_\exists$.
  Since $\EA(\sigma_\exists)$ is necessary for $\sigma_\exists$, then $A\subseteq\EA(\sigma_\exists)$.
  Since $A$ is optimal and $\EA(\sigma_\exists)$ is sufficient, we also have $\EA(\sigma_\exists)\subseteq A$. Hence the equality.
\end{proof}


\subsection{Link with strongly winning strategies}

The goal of this section is to establish a link with the notion of strongly winning strategy.
Intuitively this corresponds to the strategies that play a winning strategy whenever it is possible from the current history.
\begin{definition}[\cite{Faella09,BRS15}]
  Strategy $\sigma_\exists$ is \emph{strongly winning} when for all history $h$,
  if there exists $\sigma'_\exists$ such that $\varnothing \ne  \Out(\sigma'_\exists)  \cap h \cdot \allpaths \subseteq \spec$ then $ \Out(\sigma_\exists) \cap h \cdot \allpaths \subseteq \spec$.
  
  We will also use the notion of \emph{subgame winning} strategies (called subgame perfect in \cite{Faella09}), which are such that for all history $h$,
  if there exists $\sigma'_\exists$ such that $\Out_h(\sigma'_\exists) \subseteq \spec$ then $\Out_h(\sigma_\exists) \subseteq \spec$.
\end{definition}

\begin{lemma}[{\cite[Lem.~1]{Faella09}}]
  For every specification, there exists strongly winning and subgame winning strategies.
\end{lemma}

\begin{theorem}\label{thm:strongly-winning-optimal}
  Let $\EA(\sigma_\exists) = \spec \cup \left((\Sigma_I \cdot \Sigma_O)^\omega \setminus \Out(\sigma_\exists) \right)$.
  If strategy $\sigma_\exists$ is strongly winning for~$\spec$, then 
  $\EA(\sigma_\exists)$ is an optimal assumption for $\spec$.
  Reciprocally, if $A$ is an optimal assumption for $\spec$, then there is a strongly winning strategy $\sigma_\exists$ such that $A = \EA(\sigma_\exists)$.
\end{theorem}
\begin{proof}
  \fbox{$\Rightarrow$}
  First notice that by Lem.~\ref{lem:ea-sufficient}, $\EA(\sigma_\exists)$ is sufficient for $\sigma_\exists$ and thus sufficient for $\spec$.
  Let $A$ be an assumption which is sufficient for $\spec$, we will prove that $\EA(\sigma_\exists) \not\subset A$, which shows that $\EA(\sigma_\exists)$ is optimal.
  Let $\sigma_\exists'$ be a strategy for which $A$ is sufficient.
  If $A \setminus \EA(\sigma_\exists) = \varnothing$ then $A \subseteq \EA(\sigma_\exists$ which shows the property.
  Otherwise there exists $w \in A \setminus \EA(\sigma_\exists)$.
    Since $w \not\in \EA(\sigma_\exists)$ and $\spec \subseteq  \EA(\sigma_\exists)$, $w\not\in \spec$, i.e. $w$ is losing.
    Since $w \not\in \EA(\sigma_\exists)$ and $\allpaths \setminus \Out(\sigma_\exists) \subseteq  \EA(\sigma_\exists)$, $w\in \Out(\sigma_\exists)$, i.e. it is an outcome of $\sigma_\exists$.
    Since $A \cap \Out(\sigma'_\exists) \subseteq \spec$ and $w \in A \setminus \spec$, $w\not\in \Out(\sigma'_\exists)$, i.e. it is not an outcome of $\sigma_\exists'$.
   
  Let $w_{\le k}$ be the longest prefix of $w$ that is compatible with $\sigma_\exists'$.
  Since $\sigma_\exists$ is strongly winning and $w$ is an outcome of $\sigma_\exists$ which is losing, for all strategies $\sigma''_\exists$, either $w_{\le k} \cdot \allpaths \cap \Out(\sigma''_\exists) = \varnothing$ or $w_{\le k} \cdot \allpaths \cap \Out(\sigma''_\exists) \setminus \spec \ne \varnothing$.
  Since $w_{\le k}$ is compatible with $\sigma'_\exists$, $w_{\le k} \cdot \allpaths \cap \Out(\sigma'_\exists) \ne \varnothing$ and therefore there is an outcome~$w'$  of $\sigma_\exists'$ which is losing.
  Since $A$ is sufficient for $\sigma_\exists'$, $w' \not\in A$.
  Note that $w'$ is not an outcome of $\sigma_\exists$: $w'_{k+1} = \sigma'_\exists(w_{\le k}) \ne \sigma_\exists(w_{\le k})$.
  Hence, $w\in \allpaths \setminus \Out(\sigma_\exists) \subseteq \EA(\sigma_\exists)$.
  Therefore $w' \in \EA(\sigma_\exists) \setminus A$ which proves $\EA(\sigma_\exists) \not\subset A$.
  \medskip

  \fbox{$\Leftarrow$}
  Let $A$ be an optimal assumption for $\spec$ and $\sigma_\exists$ a strategy for which $A$ is sufficient.
  Note that by Corollary.~\ref{cor:optimal-ea}, $A \subseteq \EA(\sigma_\exists)$.
  We show that $\sigma_\exists$ is strongly winning.
  Let $h$ be a history such that there is $\sigma'_\exists$ such that $\varnothing \ne  \Out(\sigma'_\exists)  \cap h \cdot \allpaths \subseteq (\spec)$.
  We prove that $\Out(\sigma_\exists) \cap h \cdot \allpaths \subseteq \spec$ which shows the result.

  If $\Out(\sigma_\exists) \cap h \cdot \allpaths = \varnothing$ then the property is obviously satisfied.
  Otherwise there is $w \in \Out(\sigma_\exists) \cap h \cdot \allpaths$.

  Consider the strategy $\sigma_\exists\switch{h}{\sigma'_\exists}$ that plays according to $\sigma_\exists$ and when $h$ is reached shifts to $\sigma'_\exists$.
  Formally, given a history $h'$:
  \[\sigma_\exists\switch{h}{\sigma'_\exists} = \left\{\begin{array}{ll}
  \sigma'_\exists(h') & \textrm{if } h \text{ is a prefix of } h' \\
  \sigma_\exists(h') & \textrm{otherwise}
  \end{array}\right.\]
  Since $h$ is compatible with $\sigma'_\exists$ and $\Out(\sigma'_\exists)  \cap h \cdot \allpaths \subseteq \spec$, we also have $\sigma_\exists\switch{h}{\sigma'_\exists}\cap h \cdot \allpaths \subseteq \spec$.
  Moreover all outcomes that are not in $h \cdot \allpaths$ are compatible with $\sigma_\exists$.
  Hence $\EA(\sigma_\exists) \cup h \cdot \allpaths$ is sufficient for $\sigma_\exists\switch{h}{\sigma'_\exists}$.
  By optimality of assumption $\EA(\sigma_\exists)$, $\EA(\sigma_\exists) \not\subset \EA(\sigma_\exists) \cup h \cdot \allpaths$.
  Hence $h \cdot \allpaths \subseteq \EA(\sigma_\exists)$.
  Since $\EA(\sigma_\exists)$ is sufficient for $\sigma_\exists$, $\Out(\sigma_\exists) \cap h \cdot \allpaths \subseteq \spec$, which shows the result.
\end{proof}

\begin{example}\label{ex:general}
  Consider the specification $\spec = (i_1 \Rightarrow \X o_1 \land i_2 \Rightarrow \X o_2) \U o_2$, for which a B\"uchi automaton is given in \figurename~\ref{fig:i1Xo1Fo2}.
  There is no winning strategy in this game, since if the input is always $i_1$ then controller must reply by $o_1$, by the first member of the conjunction, and this prevents the third member from being satisfied.
  However, if the current history is of the form $(i_1 \cdot o_1)^* \cdot i_2$, then controller has a winning strategy which consists in replying $o_2$ and continue to imitate the inputs.
  Strongly winning strategies must therefore present this behaviour for all $(i_1 \cdot o_1)^* \cdot i_2$ that are compatible with it.

  A strongly winning strategy~$\sigma_\exists$ is the one that if the first input is $i_2$ then it plays the winning strategy we described and otherwise the first input is $i_1$ and it always play $o_2$.
  This is a strongly winning strategy since for histories beginning with $i_2$ it is winning and for any other history compatible with $\sigma_\exists$ there is no winning strategy.
  The assumption corresponding to this strategy is $\EA(\sigma_\exists) = i_2 \cdot \Sigma_O \cdot \allpaths + i_1 \cdot (\Sigma_O \cdot \Sigma_I)^* \cdot o_1 \cdot \allpaths $.
  This is indeed an optimal assumption, but it may not be what we would expect because the expression $i_1 \cdot (\Sigma_O \cdot \Sigma_I)^* \cdot o_1$ is an assumption which talks about the controller rather than the environment.
  A controller that falsifies the assumption would then be considered correct.
  Instead of this, we would prefer an assumption which talks only about the environment.
  This motivates the search for nonrestrictive assumptions.
\end{example}

\begin{figure}[h]
  \centering{
    \begin{tikzpicture}
      \draw (0,0) node[draw,minimum size=6mm] (S0) {$s_0$};
      \draw (2,-0.8) node[draw,circle] (S1) {$s_1$};
      \draw (2,0.8) node[draw,circle] (S2){$s_2$};
      \draw (4,0.8) node[draw,double,minimum size=6mm] (S3) {$s_3$};
      \draw (4,-0.8) node[draw,minimum size=6mm] (S4) {$s_4$};
      \draw (6,0.8) node[draw,double,circle] (S5) {$s_5$};
      \draw (6,-0.8) node[draw,circle] (S6) {$s_6$};

      \draw [-latex'] (-1,0) -- (S0);
      \draw [-latex'] (S0) edge[bend left] node[above] {$i_1$} (S1);
      \draw [-latex'] (S0) edge[bend left] node[above] {$i_2$} (S2);
      \draw [-latex'] (S1) edge[bend left] node[below] {$o_1$} (S0);
      \draw [-latex'] (S1) -- node[above] {$o_2$} (S4);
      \draw [-latex'] (S2) -- node[above] {$o_1$} (S4);
      \draw [-latex'] (S2) -- node[above] {$o_2$} (S3);
      \draw [-latex'] (S3) edge[bend left] node[above] {$\Sigma_I$} (S5);
      \draw [-latex'] (S5) edge[bend left] node[above] {$\Sigma_O$} (S3);
      \draw [-latex'] (S4) edge[bend left] node[above] {$\Sigma_I$} (S6);
      \draw [-latex'] (S6) edge[bend left] node[above] {$\Sigma_O$} (S4);
    \end{tikzpicture}
  }
  \caption{B\"uchi automaton for specification~$(i_1 \Rightarrow \X o_1 \land i_2 \Rightarrow \X o_2) \U o_2$.}
  \label{fig:i1Xo1Fo2}
\end{figure}
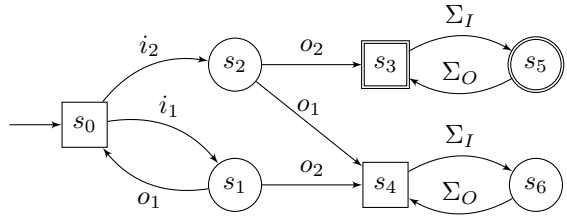


\subsection{Infinity of optimals}

\begin{theorem}\label{thm:infinity-optimal}
  There is a specification for which there are an infinite number of optimal assumptions and an infinite number of optimal ensurable assumptions.
\end{theorem}
\begin{proof}
Consider the game of \figurename~\ref{fig:infinity-general}.
In this game there are an infinite number of strongly winning strategies.
They must all play $o_2$ in $s_5$ but have the choice of how long to stay in $s_2$.
We write $\sigma_\exists^n$ the strategy that plays $o_1$, $n$ times before playing $o_2$ (note that we could also consider strategies that depend on the choice of input in $s_1$, but this will not be necessary here).
The sufficient hypothesis for $\sigma_\exists^n$ is $\EA(\sigma_\exists^n) = \allpaths \setminus (\Sigma_I \cdot o_1)^n \cdot \Sigma_I \cdot o_2 \cdot i_1 \cdot \Sigma_O \cdot \allpaths$.
They are incomparable and since $\sigma_\exists^n$ are strongly winning they are all optimal. 
This shows that there is an infinite number of optimal assumptions.
Note that these assumptions are ensurable and therefore there also is an infinite number of ensurable-optimal assumptions.
\end{proof}
\begin{figure}[h]
  \centering{
    \begin{tikzpicture}[xscale=1.5,yscale=1.5]
      \draw (0,0) node[draw,minimum size=6mm] (S1) {$s_1$};
      \draw (1,0) node[draw,circle] (S2){$s_2$};
      \draw (2,0) node[draw,minimum size=6mm] (S3) {$s_3$};
      \draw (3,0) node[draw,circle] (S5) {$s_5$};
      \draw (2,-1) node[draw,circle] (S7) {$s_7$};
      \draw (3,-1) node[draw,minimum size=6mm] (S8) {$s_8$};
      \draw (5,0) node[draw,circle,double] (S9) {$s_9$};
      \draw (4,0) node[draw,minimum size=6mm,double] (S10) {$s_{10}$};

      \draw [-latex'] (-1,0) -- (S1);
      \draw [-latex'] (S1) -- node[below] {$\Sigma_I$} (S2);
      \draw [-latex'] (S2) -- node[above] {$o_2$} (S3);
      \draw [-latex'] (S2) edge[bend right] node[above] {$o_1$} (S1);
      \draw [-latex'] (S3) -- node[right] {$i_1$} (S7);
      \draw [-latex'] (S3) -- node[above] {$i_2$} (S5);
      \draw [-latex'] (S5) -- node[right] {$o_1$} (S8);
      \draw [-latex'] (S5) -- node[above] {$o_2$} (S10);
      \draw [-latex'] (S7) edge[bend left] node[above] {$\Sigma_O$} (S8);
      \draw [-latex'] (S8) edge[bend left] node[below] {$\Sigma_I$} (S7);
      \draw [-latex'] (S9) edge[bend left] node[below] {$\Sigma_I$} (S10);
      \draw [-latex'] (S10) edge[bend left] node[above] {$\Sigma_O$} (S9);
    \end{tikzpicture}
  }
  \caption{A specification with an infinity of optimal assumptions.}
  \label{fig:infinity-general}
\end{figure}
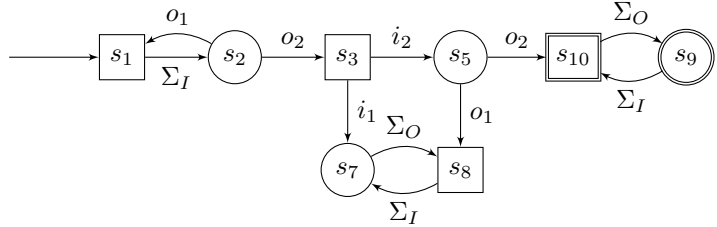

\subsection{Scenarios}

\begin{theorem}\label{thm:ea-scenario}
  If $\sigma_\exists$ is subgame winning strategy for $\spec$, then $\EA([S \rightarrow \sigma_\exists])$ is optimal for $\spec$ and $S$. 
  We recall that $\EA([S \rightarrow \sigma_\exists]) = \spec \cup \left((\Sigma_I \cdot \Sigma_O)^\omega \setminus \Out([S \rightarrow \sigma_\exists]) \right)$ and \[[S \rightarrow \sigma_\exists] (h) = \left\{\begin{array}{ll}
  w_{h+1} & \textrm{if there is } w \in S \text{ such that } h \text{ is a prefix of } w\\
  \sigma_\exists(h) & \textrm{otherwise.}
  \end{array}\right.\]
\end{theorem}
\begin{proof}
  The assumption $\EA([S \rightarrow \sigma_\exists])$ is sufficient for $[S \rightarrow \sigma_\exists]$ thanks to Lem.~\ref{lem:ea-sufficient}, moreover $[S \rightarrow \sigma_\exists]$ is compatible with $S$ by Lem.~\ref{lem:scenario-compatible}.

  Assume there is $H$ and $\sigma'_\exists$ such that $H$ is sufficient for $\sigma'_\exists$, $\sigma'_\exists$ compatible with $S$ and $H$ contains $\EA([S \rightarrow \sigma_\exists])$.
  Let $w \in H$ we want to prove that $w \in \EA([S \rightarrow \sigma_\exists])$ which will show that $H \subseteq \EA([S \rightarrow \sigma_\exists])$ and thus that $\EA([S \rightarrow \sigma_\exists])$ is optimal.

  Assume towards a contradiction that $w \not\in \EA([S \rightarrow \sigma_\exists])$ and therefore $w \not\in \spec$ and $w\in \Out(\sigma_\exists)$.
  Since $H$ is sufficient for $\sigma'_\exists$, $w \not\in \Out(\sigma'_\exists)$.
  Let $i$ be the first index such that $w_{i} \ne \sigma'_\exists(w_{\le i-1})$.
  Since  $\sigma_\exists$ and $\sigma'_\exists$ are compatible with $S$, $w_{\le i-1}$ is not a prefix of a word in $S$.
  As a consequence we have that $\Out([S\rightarrow \sigma_\exists]) \cap w_{\le i} \cdot \allpaths = \Out_{w_{\le i}}(\sigma_\exists)$.

  We have that $w_{\le i-1} \cdot \sigma'_\exists(w_{\le i-1}) \cdot \allpaths \subseteq \EA(\sigma_\exists)$ since these are not outcomes of $\sigma_\exists$.
  Since $H$ includes $\EA(\sigma_\exists)$, and $H$ is sufficient for $\sigma_\exists$, $\sigma'_\exists$ is winning from $w_{\le i-1}$ (i.e. $\Out(\sigma'_\exists) \cap w_{\le i-1} \cdot \Sigma_O \cdot \allpaths \subseteq \spec$).
  As $\sigma_\exists$ is subgame winning, $\Out_{w_{\le i}}(\sigma_\exists) \subseteq \spec$ and in particular $w \in \spec$.
  This shows that $\EA([S \rightarrow \sigma_\exists])$ is optimal.
\end{proof}

\subsection{Generalisation}

Assume now we are given a sufficient assumption $H$ and want to generalise it, that is find $H'$ optimal and such that $H \subseteq H'$.
We compute $\sigma_\exists$ winning for $H \Rightarrow \spec$ (i.e. such that $\Out(\sigma_\exists)\cap H \subseteq \spec$ and $\sigma'_\exists$ subgame winning for $\spec$.
We then define $\sigma'_\exists[H \setminus W \rightarrow \sigma_\exists]$ to be the function that maps $h$ to $\sigma'_\exists(h)$ if $h$ is not a prefix of a word $w \in H$ or $h \in W = \{ h \mid \Out(\sigma'_\exists)\cap h\cdot \allpaths \subseteq \spec\}$, and maps $h$ to $\sigma_\exists(h)$ otherwise.
\begin{lemma}\label{lem:generalisation}
  If $\Out(\sigma_\exists) \cap H \subseteq \spec$ and $\sigma'_\exists$ is subgame winning for $\spec$, then $\EA(\sigma'_\exists[H \setminus W \rightarrow \sigma_\exists])$ is an optimal assumption for $\spec$ and contains $H$.
\end{lemma}
\begin{proof}
  First, we show that $\sigma'_\exists[H \setminus W \rightarrow \sigma_\exists]$ is winning for $H \Rightarrow \spec$.
  Let $w \in H \cap \Out(\sigma'_\exists[H \setminus W \rightarrow \sigma_\exists])$.
  We have that for all $i \in \mathbb{N}$, $w_{\le 2 \cdot i+1}$ is prefix of a word in $H$, so $\sigma'_\exists[H \setminus W \rightarrow \sigma_\exists](w_{\le 2 \cdot i+1})=\sigma_\exists(w_{\le 2 \cdot i +1})$.
  Therefore $w \in \Out(\sigma_\exists)$ and since $\Out(\sigma_\exists)\cap H \subseteq \spec$, $w\in \spec$.

  We now show that $\sigma'_\exists[H \setminus W \rightarrow \sigma_\exists]$ is strongly winning for $\spec$.
  Let $h$ be a history such that there exists $\sigma''_\exists$ and $\varnothing \ne h \cdot \allpaths \subseteq \spec$.
  Since $\sigma'_\exists$ is subgame winning, we know that $\Out_h(\sigma'_\exists) \subseteq \spec$.
  This implies that $h \in W$, then $\sigma'_\exists[H \setminus W \rightarrow \sigma_\exists]$ plays according to $\sigma'_\exists$ for the rest of the play which means $\Out_h(\sigma'_\exists[H \setminus W \rightarrow \sigma_\exists]) \subseteq \spec$.

  Then $\sigma'_\exists[H \setminus W \rightarrow \sigma_\exists]$ is subgame winning for $\spec$ and winning for $H \Rightarrow \spec$.
  By Thm.~\ref{thm:strongly-winning-optimal}, $\EA(\sigma'_\exists[H \setminus W \rightarrow \sigma_\exists])$ is an optimal assumption for $\spec$.
  Since it is winning for $H \Rightarrow \spec$, $\Out(\sigma'_\exists[H \setminus W \rightarrow \sigma_\exists]) \cap H \subseteq \spec$.
  This means that  $\EA(\sigma'_\exists[H \setminus W \rightarrow \sigma_\exists]) = \spec \cup \left((\Sigma_I \cdot \Sigma_O)^\omega \setminus \Out(\sigma'_\exists[H \setminus W \rightarrow \sigma_\exists]) \right)$ contains $H$.
\end{proof}

\section{Ensurable assumptions}\label{sec:ensurable}

\subsection{Necessary and sufficient non-restrictive assumptions}

In this section, we show properties of assumption that are not restrictive.
As we have seen in Lem.~\ref{lem:ensurable-nonrestrictive}, this coincide with ensurable assumptions for $\omega$-regular objectives.

\def\Doomed{\textsf{Doomed}}
Given a strategy $\sigma_\exists$ of \eve, the word~$w$ is \emph{doomed for} $\sigma_\exists$ if
there is an index $k$ such that one outcome of $\sigma_\exists$ has prefix $w_{\le k}$ and all outcome of $\sigma_\exists$ that have prefix $w_{\le k}$ do not satisfy $\spec$.
We write $\Doomed(\sigma_\exists)$ for the set of words that are doomed for $\sigma_\exists$ i.e. $\Doomed(\sigma_\exists) = \{ w \mid \exists k\in 2\cdot \mathbb{N}.\ \Out(\sigma_\exists) \cap w_{\le k} \cdot \allpaths \ne \varnothing \text{ and } \Out(\sigma_\exists) \cap w_{\le k} \cdot \allpaths \cap \spec = \varnothing \}$.
We consider the assumption $\EA^-(\sigma_\exists) = \EA(\sigma_\exists) \setminus \Doomed(\sigma_\exists)$.

\begin{lemma}\label{lem:ea-minus-sufficient}
  Let $\sigma_\exists$ be a strategy, we have the following properties:
  \begin{enumerate}
  \item\label{ea-suf} $\EA^-(\sigma_\exists)$ is sufficient for $\sigma_\exists$, and nonrestrictive;
  \item\label{ea-nec} for all assumption~$A$
    sufficient for $\sigma_\exists$ and not output-restrictive, we have that $A \subseteq \EA^-(\sigma_\exists)$.
  \end{enumerate}
\end{lemma}
\begin{proof}
  \textbf{\ref{ea-suf}.}
  The fact that $\EA^-(\sigma_\exists)$ is sufficient for $\sigma_\exists$ is a consequence of the fact that it is included in $\EA(\sigma_\exists)$ which is sufficient for $\sigma_\exists$ (Lem.~\ref{lem:ea-sufficient}).
  Let us show that $\EA^-(\sigma_\exists)$ is not output-restrictive.
  Let $w \in \EA^-(\sigma_\exists)$,
  we show that for all $k$, and all strategy $\sigma'_\exists$, if $w_{\le k}$ is prefix of some outcome of $\sigma'_\exists$ then $w_{\le k}$ can be completed in a word of $\EA^-(\sigma_\exists) \cap \Out(\sigma'_\exists)$.

  Assume that $w_{\le k}$ is prefix of an outcome of $\sigma'_\exists$.
  Since $w$ is not doomed for $\sigma_\exists$, either $w_{\le k}$ is not a prefix of an outcome of $\sigma_\exists$, 
  or there is an outcome~$w'$ of $\sigma_\exists$ that has prefix $w_{\le k}$ and satisfies $\spec$.
  \begin{itemize}
  \item If $w_{\le k}$ is not a prefix of an outcome of $\sigma_\exists$, then any suffix of $w_{\le k}$ belongs to $\EA^-(\sigma_\exists)$, by selecting an arbitrary outcome of $\sigma'_\exists$ after $w_{\le k}$ we obtain a word of $\EA^-(\sigma_\exists) \cap \Out(\sigma'_\exists)$.
  \item Otherwise let $w'$ outcome of $\sigma_\exists$ that has prefix $w_{\le k}$ and satisfies $\spec$.
    The word $w'$ is not doomed: any prefix of $w'$ is prefix of a word (in particular $w'$) which is an outcome of $\sigma_\exists$ and satisfy $\spec$.
    Let $\sigma_\forall$ be a strategy of \adam such that $\Out(\sigma_\exists,\sigma_\forall) = w'$.
    Let $w'' = \Out(\sigma'_\exists,\sigma_\forall)$.
    \begin{itemize}
    \item If $w'' = w'$.
      The word $w'$ belongs to $\EA(\sigma_\exists)$ and it is not doomed, hence $w''$ belongs to $\EA^-(\sigma_\exists)$.
    \item Otherwise, we have that $w'' \not\in \Out(\sigma_\exists)$ and therefore $w'' \in \EA(\sigma_\exists)$.
      Lets show that $w''$ is not doomed.
      Let $k$ be the smallest index such that $w''_k \ne w'_k$ and let $k'$ be another index.
      If $k' < k$ then $w''_{\le k'} < w'_{\le k'}$ and since $w'$ is not doomed, there is an outcome of $\sigma_\exists$ that have prefix $w''_{\le k'}$ and satisfy $\spec$.
      If $k' \ge k$ then $w''_{\le k'}$ is prefix of no outcome of $\sigma_\exists$.
      Therefore $w''$ is not doomed for $\sigma_\exists$ and $w'' \in\EA^-(\sigma_\exists)$.
    \end{itemize}
  \end{itemize}
  This shows that $\EA^-(\sigma_\exists)$ is not output-restrictive.
  \medskip 
  
  \textbf{\ref{ea-nec}.}
  Lets now assume that $A$ is sufficient for $\sigma_\exists$ and not output-restrictive.
  Let $w\in A$, since $A$ is sufficient we have that $w \in \EA(\sigma_\exists)$ by (2).
  Let us show that $w \not\in \Doomed(\sigma_\exists)$.

  Towards a contradiction assume $w\in \Doomed(\sigma_\exists)$ and let $k$ be the first index such that for all outcome~$w'$ of $\sigma_\exists$ that have prefix $w_{\le k}$, $w' \not\models \spec$.
  Since $A$ is not output restrictive, $A \cap w_{\le k}\cdot \allpaths \ne \varnothing $, and $\Out(\sigma_\exists) \cap w_{\le k}\cdot \allpaths \ne \varnothing$, we have that $A \cap \Out(\sigma_\exists) \cap w_{\le k}\cdot \allpaths \ne \varnothing$.
  Therefore there is a word $w'$ that has prefix $w_{\le k}$ and belongs to $A \cap \Out(\sigma_\exists)$.
  By definition of $w_{\le k}$, the word $w'$ does not satisfy $\spec$.
  This contradicts the fact that $A$ is sufficient for $\spec$.
  
  This proves that $w$ is not doomed for $\sigma_\exists$ and therefore belongs to $\EA^-(\sigma_\exists)$ which shows that $A \subseteq \EA^-(\sigma_\exists)$.
\end{proof}

\begin{example}
  For the strategy $\sigma_\exists$ we defined in example~\ref{ex:general}, the set of doomed histories is $i_1 \cdot o_2 \cdot \allpaths$.
  Then $\EA^-(\sigma_\exists)$ is $i_2 \cdot \Sigma_O \cdot (\Sigma_I \cdot \Sigma_O)^\omega$ which is nonrestrictive.
  This assumption describes better than $\EA(\sigma_\exists)$ the assumptions on the environment necessary to win.
  However it is not optimal among nonrestrictive assumptions, and we will now characterise the strategies for which $\EA^-(\sigma_\exists)$ is optimal.
\end{example}

\subsection{Link with non-dominated strategies}

We use the notion of weak dominance classical in game theory.
Intuitively a strategy dominates another one if it performs at least as well against any strategy of the environment.
\begin{definition}[\cite{BRS14}]
  Strategy $\sigma_\exists$ is \emph{very weakly dominated} from history~$h$ by strategy $\sigma'_\exists$ if for all strategy $\sigma_\forall$ of \adam, $\Out_h(\sigma_\exists, \sigma_\forall) \in \spec \Rightarrow \Out_h(\sigma'_\exists, \sigma_\forall) \in \spec$.
  It is \emph{weakly dominated} from $h$ by $\sigma'_\exists$ if moreover $\sigma'_\exists$ is not very weakly dominated by $\sigma_\exists$ from $h$.
  A strategy is said \emph{non-dominated} if there is no strategy that weakly-dominates it from the empty history~$\varepsilon$.
  A strategy is \emph{non-subgame-dominated} if there is no strategy that weakly-dominates it from any history~$h$.
\end{definition}


We draw a link between optimal assumptions and non-dominated strategies.
\begin{lemma}\label{lem:subset-implies-weakly-dominated}
  If
  $\EA^-(\sigma_\exists) \subseteq \EA^-(\sigma'_\exists)$ 
  then $\sigma_\exists$ is very weakly dominated by $\sigma'_\exists$.
\end{lemma}
\begin{proof}
  Let $\sigma_\forall$ be a strategy of \adam.
  Let $w = \Out(\sigma_\exists,\sigma_\forall)$ and  $w' = \Out(\sigma'_\exists,\sigma_\forall)$.
  We aim at showing $w \models \spec \Rightarrow w' \models \spec$.
  If $w \not\models \spec$ or $w = w'$ the implication is obvious.
  Assume now $w\models \spec$ and $w \ne w'$, we will prove that $w'\models \spec$.
  Since $w\ne w'$, let $k$ be the greatest index such that $w_{\le k} = w'_{\le k}$.
  Note that $w_k$ is not controlled by \adam, since the $w$ and $w'$ result from the same strategy of \adam and $w_{k+1} \ne w'_{k+1}$, therefore it is controlled by \eve.


  Assume towards a contradiction that $w'$ does not satisfy $\spec$.
  Then $w' \in \Out(\sigma'_\exists) \setminus \spec$ so $w' \not\in \EA(\sigma'_\exists)$ and $w' \not\in \EA^-(\sigma'_\exists)$.
  Using the hypothesis that $\EA^-(\sigma_\exists) \subseteq \EA^-(\sigma'_\exists)$ this means that $w' \not\in \EA^-(\sigma_\exists)$.
  Therefore either $w' \in \Out(\sigma_\exists) \setminus \spec $ or $w'$ is doomed for $\sigma_\exists$.
  \begin{itemize}
  \item If $w' \in \Out(\sigma_\exists) \setminus \spec$, 
    the case is solved as for the case $\EA(\sigma_\exists) \subseteq \EA(\sigma'_\exists)$.
  \item Otherwise $w'$ is doomed for $\sigma_\exists$.
    Let $k'$ be such that $w'_{\le k'}$ is an outcome of $\sigma_\exists$ and all outcome of $\sigma_\exists$ that have prefix $w_{\le k}$ do not satisfy $\spec$.
    If $k' \le k$ then since $w$ has $w'_{\le k'}$ for prefix and is an outcome of $\sigma_\exists$, it does satisfy $\spec$ which is a contradiction.
    Otherwise $k' > k$.
    Since we mentioned previously that $w_k$ is controlled by \eve, the fact that $w'_{\le k'}$ is prefix of an outcome of $\sigma_\exists$ means that $\sigma_\exists(w_{\le k}) = w'_{k+1}$ which contradicts the fact that $k$ is the greatest index such that $w_{\le k} = w'_{\le k}$.
  \end{itemize}
  Therefore $w'$ satisfies $\spec$ and the implication holds.
\end{proof}

\begin{example}
  Consider the game in \figurename~\ref{fig:counter-ea}.
  There are two kind strategies: $\sigma_\exists$ which starts by playing $o_1$ and $\sigma'_\exists$ which plays $o_2$.
  We have that $\sigma'_\exists$ weakly dominates $\sigma_\exists$.
  The assumption $\EA(\sigma_\exists)$ necessary to $\sigma_\exists$ is $\Sigma_I \cdot o_2 \cdot \allpaths$ while the assumption necessary to $\sigma'_\exists$ is $\EA(\sigma'_\exists) = \Sigma_I \cdot \Sigma_O \cdot i_1 \cdot \allpaths$.
  $\Doomed(\sigma'_\exists) = \Sigma_I \cdot o_2 \cdot i_2 \cdot \Sigma_O \cdot \allpaths$ while $\Doomed(\sigma_\exists) = \allpaths$.
  So while $\EA(\sigma_\exists)$ and $\EA(\sigma'_\exists)$ are incomparable, we indeed have $\EA^-(\sigma_\exists) \subset \EA^-(\sigma'_\exists)$.
\end{example}

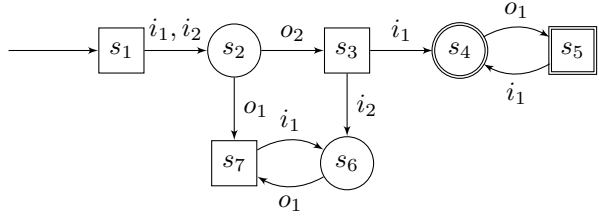
\begin{figure}[h]
  \centering{
    \begin{tikzpicture}[xscale=1.5,yscale=1.5]
      \draw (0,0) node[draw,minimum size=6mm] (S1) {$s_1$};
      \draw (1,0) node[draw,circle] (S2){$s_2$};
      \draw (2,0) node[draw,minimum size=6mm] (S3) {$s_3$};
      \draw (3,0) node[draw,circle,double] (S4){$s_4$};
      \draw (4,0) node[draw,minimum size=6mm,double] (S5) {$s_5$};
      \draw (2,-1) node[draw,circle] (S6) {$s_6$};
      \draw (1,-1) node[draw,minimum size=6mm] (S7) {$s_7$};

      \draw [-latex'] (-1,0) -- (S1);
      \draw [-latex'] (S1) -- node[above] {$i_1,i_2$} (S2);
      \draw [-latex'] (S2) -- node[above] {$o_2$} (S3);
      \draw [-latex'] (S2) -- node[right] {$o_1$} (S7);
      \draw [-latex'] (S3) -- node[above] {$i_1$} (S4);
      \draw [-latex'] (S3) -- node[right] {$i_2$} (S6);
      \draw [-latex'] (S4) edge[bend left] node[above] {$o_1$} (S5);
      \draw [-latex'] (S5) edge[bend left] node[below] {$i_1$} (S4);
      \draw [-latex'] (S6) edge[bend left] node[below] {$o_1$} (S7);
      \draw [-latex'] (S7) edge[bend left] node[above] {$i_1$} (S6);
    \end{tikzpicture}
  }
  \caption{A illustration of a game where $\Doomed$ makes a difference between restrictive and non-restrictive assumptions.}
  \label{fig:counter-ea}
\end{figure}

\begin{lemma}\label{lem:dominated-implies-doomed-subset}
  If $\sigma_\exists$ is very weakly dominated by $\sigma'_\exists$ then $\Doomed(\sigma'_\exists) \subseteq \Doomed(\sigma_\exists)$.
\end{lemma}
\begin{proof}
  Let $w$ be a word doomed for $\sigma'_\exists$.
  If $w$ is not an outcome of $\sigma_\exists$, consider the first index $k$ such that $\sigma_\exists(w_{\le k}) \ne \sigma'_\exists(w_{\le k})$.
  Let $d$ be the index from which $w$ is doomed, that is $w_{\le d}$ is prefix of an outcome of $\sigma'_\exists$ and all outcome that have $w_{\le d}$ as prefix do not satisfy $\spec$.
  Let $\sigma_\forall$ be a strategy of \adam.
  Note that if $\Out(\sigma'_\exists,\sigma_\forall)$ has prefix $w_{\le d}$ then $\Out(\sigma_\exists,\sigma_\forall)$ has prefix $w_{\le \min(d,k)}$.
  Moreover if $\Out(\sigma'_\exists,\sigma_\forall)$ has prefix $w_{\le d}$ then it does not satisfy $\spec$ and since $\sigma'_\exists$ very weakly dominates $\sigma_\exists$, $\Out(\sigma_\exists,\sigma_\forall)$ does not satisfy $\spec$.
  Therefore this means $w_{\le d}$ is prefix of an outcome of $\sigma_\exists$ and all outcomes of $\sigma_\exists$ that have this prefix do not satisfy $\spec$.
  Therefore $w$ is doomed for $\sigma_\exists$.
\end{proof}

\begin{lemma}\label{lem:weakly-dominated-implies-output-restrictive-or-subset}
  If $\sigma_\exists$ is very weakly dominated by $\sigma'_\exists$ then
  $\EA^-(\sigma_\exists) \subseteq \EA^-(\sigma'_\exists)$.
\end{lemma}
\begin{proof}
  Let $w \in \EA^-(\sigma_\exists) = \EA(\sigma_\exists) \setminus \Doomed(\sigma_\exists)$.
  Since $\sigma_\exists$ is very weakly dominated by $\sigma'_\exists$, we have by Lem.~\ref{lem:dominated-implies-doomed-subset} that $\Doomed(\sigma'_\exists) \subseteq \Doomed(\sigma_\exists)$.
  Therefore $w \not\in \Doomed(\sigma'_\exists)$.

  Assume towards a contradiction that $w\not\in \EA(\sigma'_\exists)$.
  Then $w\not\in \spec$ and $w\in \Out(\sigma'_\exists)$.
  Let $w_{\le k}$ be the longest such that there is an outcome of $\sigma_\exists$ that has this prefix.
  Note that $k$ is odd, because $\sigma_\exists$ does not take the decision at odd positions.
  Since $w\not\in\Doomed(\sigma_\exists)$ there is an outcome~$w'$ of $\sigma_\exists$ that has $w_{\le k}$ for prefix and that is winning.
  We define the strategy $\sigma_\forall$ of the environment such that:
  \[
  \sigma_\forall(h) = 
  \left\{
  \begin{array}{r l} 
    w_{i+1} & \text{ if } h = w_{\le i} \text{ for some }i \\
    w'_{i+1} & \text{ if } h = w'_{\le i} \text{ for some }i \\
    \sigma_\forall'(h) & \text{ otherwise}
  \end{array}
  \right.
  \]
  where $\sigma_\forall'$ is some arbitrary strategy that we have fixed.
  Note that this is well defined because if $w_{\le i} = w'_{\le i}$ and $i$ is even then $i \le k$ and therefore $w_{i+1} = w'_{i+1}$.
  We have that $\Out(\sigma_\exists,\sigma_\forall) = w'$ and $\Out(\sigma_\exists',\sigma_\forall) = w$.
  Thus $\Out(\sigma_\exists,\sigma_\forall) \in \spec$ and $\Out(\sigma_\exists',\sigma_\forall) \not\in \spec$ which contradicts the fact that $\sigma_\exists$ is very weakly dominated by $\sigma'_\exists$.
  Hence $w \in \EA(\sigma'_\exists) \setminus \Doomed(\sigma'_\exists) = \EA^-(\sigma'_\exists)$.
\end{proof}

\begin{theorem}\label{thm:non-dominated-nonrestrictive}
  Let $\spec$ be an $\omega$-regular specification.
  If $\sigma_\exists$ is a non-dominated strategy for~$\spec$, then $\EA^-(\sigma_\exists)$ is ensurable optimal for $\spec$.
  Reciprocally if $A$ is an ensurable optimal assumption for $\spec$, then there is $\sigma_\exists$ a non-dominated strategy for~$\spec$ such that $A = \EA^-(\sigma_\exists)$.
\end{theorem}
\begin{proof}
  \fbox{$\Rightarrow$}
  Let $\sigma_\exists$ be a non-dominated strategy.
  Note already that by Lem.~\ref{lem:ea-minus-sufficient}, the assumption $\EA^-(\sigma_\exists)$ is sufficient for $\spec$ and not output-restrictive.
  We now prove it is optimal.
  Let $A$ be a nonrestrictive assumption sufficient for $\spec$.
  There is a strategy $\sigma'_\exists$ such that $A$ is sufficient for $\sigma'_\exists$.
  By Lem.~\ref{lem:ea-sufficient}, $A \subseteq \EA(\sigma'_\exists)$.
  As $\sigma_\exists$ is not weakly dominated by $\sigma'_\exists$, either:
  \begin{itemize}
    \item $\sigma_\exists$ is not very weakly dominated by $\sigma'_\exists$.
      Then by Lem.~\ref{lem:subset-implies-weakly-dominated} $\EA^-(\sigma_\exists) \not\subseteq \EA^-(\sigma'_\exists)$.
    \item or $\sigma_\exists$ very weakly dominates $\sigma'_\exists$ then by Lem.~\ref{lem:weakly-dominated-implies-output-restrictive-or-subset}, $\EA^-(\sigma'_\exists) \subseteq \EA^-(\sigma_\exists)$.
  \end{itemize}
  Therefore $\EA^-(\sigma_\exists) \not\subset \EA^-(\sigma'_\exists)$ and as $A \subseteq \EA^-(\sigma'_\exists)$, $\EA^-(\sigma_\exists) \not\subset A$.
  This shows that $\EA^-(\sigma_\exists)$ is optimal among nonrestrictive assumptions for $\spec$.

  \medskip

  \fbox{$\Leftarrow$}
  Let $\sigma_\exists$ be a strategy such that $\EA^-(\sigma_\exists)$ is nonrestrictive-optimal, we show that $\sigma_\exists$ is non-dominated.
  Let $\sigma'_\exists$ be a strategy which very weakly dominates $\sigma_\exists$, we prove that $\sigma_\exists$ very weakly dominates $\sigma'_\exists$, which shows that $\sigma_\exists$ is not weakly dominated.
  By Lem.~\ref{lem:weakly-dominated-implies-output-restrictive-or-subset}, $\EA^-(\sigma_\exists) \subseteq \EA^-(\sigma'_\exists)$.
  Since $\EA^-(\sigma_\exists)$ is optimal, $\EA^-(\sigma_\exists) \not\subset \EA^-(\sigma'_\exists)$.
  Therefore $\EA^-(\sigma_\exists) = \EA^-(\sigma'_\exists)$.
  By Lem.~\ref{lem:subset-implies-weakly-dominated} this implies that $\sigma'_\exists$ is very weakly dominated by $\sigma_\exists$ which shows the property.
\end{proof}


A strategy is said \emph{dominant} if it very weakly dominates all strategies.
We obtain the following from the previous theorem.
\begin{corollary}
  If $\sigma_\exists$ is dominant then $\EA^-(\sigma_\exists)$ is the unique nonrestrictive optimal assumption.
\end{corollary}
\begin{proof}
  A strategy that is dominant is non-dominated and therefore the fact that $\EA^-(\sigma_\exists)$ is nonrestrictive optimal is a consequence of Thm.~\ref{thm:non-dominated-nonrestrictive}.
  To prove uniqueness assume $A$ is a nonrestrictive optimal assumption.
  By Thm.~\ref{thm:non-dominated-nonrestrictive}, there is $\sigma'_\exists$ such that $A = \EA^-(\sigma'_\exists)$.
  Since $\sigma_\exists$ is dominant, $\sigma'_\exists$ is very weakly dominated by $\sigma_\exists$.
  By Lem.~\ref{lem:weakly-dominated-implies-output-restrictive-or-subset}, $\EA^-(\sigma'_\exists) \subseteq \EA^-(\sigma_\exists)$ and since $A$ is optimal we also have $\EA^-(\sigma_\exists) \subseteq \EA^-(\sigma'_\exists)$ which proves equality and thus uniqueness of the optimal.
\end{proof}

\subsection{Computation of optimal ensurable assumptions}

In parity games, deciding the existence of a winning strategy from a certain state can be done \NP $\cap$ \co\NP~\cite{EJ91}. 
We will show that if we have available an algorithm for solving parity games then the remaining of the operations to obtain optimal assumptions can be performed efficiently.
We first construct a representation of one arbitrary non-dominated strategy.
Our construction is based on the notion of memoryless strategies: given $\spec$ as a parity automaton, a strategy is said \emph{memoryless} if it only depends on the current state of the automaton, in other words it can be implemented with a Moore machine which has the same structure than the given automaton.

\begin{lemma}\label{lem:Moore-non-dominated}
  Given a parity automaton and a memoryless strategy $\sigma_\exists$ which ensures we are winning from each state in the winning region, we can compute in polynomial time 
  a Moore machine implementing a memoryless non-dominated strategy $\sigma'_\exists$.
\end{lemma}
\begin{proof}
  By the characterisation of Berwanger~\cite[Lem.~9]{berwanger07}, a non-dominated strategy $\sigma'_\exists$ has to be such that for all history~$h$, the value of $h$ is equal to the value of $\sigma'_\exists$ for $h$, for parity games this means:
  \begin{enumerate}
  \item if $h$ ends with a winning state, then $\sigma'_\exists$ should be winning from $h$;
  \item if $h$ ends in a state from which there is no winning path, $\sigma'_\exists$ can behave arbitrarily;
  \item if $h$ ends in another state, then there should exist one outcome of $\sigma'_\exists$ from $h$ which is wining.
  \end{enumerate}
  We define $\sigma'_\exists$ to behave like $\sigma_\exists$ in winning states, which already ensures that condition 1 and 2 are satisfied.
  For states which are not in the winning region but from which there is a winning path, we consider now that all states are controllable and compute a memoryless winning strategy.
  Such a strategy exists because from all these states there exists a winning path, and it can be computed in polynomial time because there is only one player.
  We now set $\sigma'_\exists$ to behave like this strategy on states which are not in the winning region.
  It ensures that one of its outcome is winning.
  Hence satisfies all the condition.  
  The corresponding Moore machine is obtained by removing from $\G$ the transitions that are not taken by $\sigma'_\exists$.
\end{proof}

By combining this construction with a parity automaton for $\spec$ it is possible to compute $\EA(\sigma'_\exists) = \spec \cup (\allpaths \setminus \Out(\sigma_\exists))$.

\begin{lemma}\label{lem:Moore-to-Out}
  Given a Moore machine representing a strategy $\sigma'_\exists$, we can construct in polynomial time a safety automaton recognising $\Out(\sigma_\exists)$.
\end{lemma}
\begin{proof}
  To obtain this automaton we add to the Moore machine states of the form $S\times \Sigma_I$ and a rejecting state $\bot$ which is absorbing.
  The transition function $\delta'$ is such that for all $s\in S$ and $i \in \Sigma_I$, $\delta'(s, i) = (s,i)$; if $G(s,i) = o$ and $\delta(s,i) = s'$ then $\delta'((s,i),o) = s'$; if $G(s,i) \ne o$ then $\delta'((s,i),o) = \bot$.
  This defines a safety automaton which recognises $\Out(\sigma_\exists)$.
\end{proof}

\begin{lemma}\label{lem:auto-EA}
  Given $\spec$ as a parity automaton and a strategy $\sigma'_\exists$, we can compute in polynomial time a parity automaton that recognises $\EA(\sigma'_\exists)$.
\end{lemma}
\begin{proof}
  By Lem.~\ref{lem:Moore-to-Out}, we can construct a safety automaton that recognises $\Out(\sigma_\exists)$.
  Consider the product between the automaton for $\spec$ and for $\Out(\sigma_\exists)$.
  We write $(s,t)$ for a state of this automaton where $s$ corresponds to the state for $\spec$ and $t$ for $\Out(\sigma_\exists)$.
  We set the colours to be the same than $s$ if the colour for $t$ is $0$ and the colour is $0$ otherwise.
  A word that is not in $\Out(\sigma_\exists)$ will reach $\bot$ (which has now colour $0$) at some point and therefore is accepted by our construction.
  A word that is in $\spec$ in also accepting.
  Hence this automaton recognises $\EA(\sigma'_\exists)$.
\end{proof}

\begin{lemma}\label{lem:non-dominated-strongly-winning}
  If $\sigma_\exists$ is non-subgame-dominated then $\sigma_\exists$ is subgame winning.
\end{lemma}
\begin{proof}
  Let $h$ be a history such that there is a strategy $\sigma'_\exists$ such that 
  $\Out_h(\sigma_\exists') \subseteq \spec$.
  We will prove that $\Out_h(\sigma_\exists) \subseteq \spec$ which shows that $\sigma_\exists$ is strongly winning.

  Assume $w \in \Out_h(\sigma_\exists)$, and let $\sigma_\forall$ be such that $w= \Out_h(\sigma_\exists,\sigma_\forall)$.
  Consider the strategy $\sigma_\exists\switch{h}{\sigma'_\exists}$ that plays according to $\sigma_\exists$ and when $h$ is reached shifts to $\sigma'_\exists$.
  Formally, given a history $h'$:
  \[\sigma_\exists\switch{h}{\sigma'_\exists} = \left\{\begin{array}{ll}
  \sigma'_\exists(h') & \textrm{if } h \text{ is a prefix of } h' \\
  \sigma_\exists(h') & \textrm{otherwise}
  \end{array}\right.\]
  Note that we already used this construction in the proof of Thm.~\ref{thm:strongly-winning-optimal}.

  Strategy $\sigma_\exists$ is weakly dominated by $\sigma_\exists\switch{h}{\sigma'_\exists}$ from $h$ 
  because the second strategy is winning from $h$.
  Then, since $\sigma_\exists$ is non-subgame-dominated, it weakly dominates $\sigma_\exists\switch{h}{\sigma'_\exists}$ from $h$.
  We have that $\Out_h(\sigma'_\exists,\sigma_\forall)$ is winning, therefore $\sigma_\exists\switch{h}{\sigma'_\exists}$ is winning against $\sigma_\forall$.
  Since $\sigma_\exists$ weakly dominates $\sigma_\exists\switch{h}{\sigma'_\exists}$, it is also winning against $\sigma_\forall$ and $w \in \spec$.
\end{proof}

As strongly non-dominated strategy is also strongly winning~(Lem.~\ref{lem:non-dominated-strongly-winning}), Lem.~\ref{lem:auto-EA} already gives an algorithm to compute a (not necessarily ensurable) optimal assumption.
\begin{corollary}
  Given an oracle to compute a memoryless winning strategy in a parity game, we can compute efficiently an optimal assumption as a disjunction of parity automata.
\end{corollary}

Now to remove doomed histories, we remove transitions going to states from which there is no winning path.
\begin{lemma}\label{lem:auto-Doomed}
  Given a parity automaton and a strategy~$\sigma_\exists$, we can compute in polynomial time a safety automaton recognising $\allpaths \setminus \Doomed(\sigma_\exists)$.
\end{lemma}
\begin{proof}
  Consider the product between the parity automaton and the Moore machine implementing $\sigma_\exists$.
  We consider $D$ the set of states from which there is no winning path.
  This set can be computed in polynomial time (it can be seen as the winning states of a one player parity game).
  To obtain the desired safety automaton, we replace transitions leading to state in $D$ to a state $\bot$ that is rejecting.

  Let $w\in\Doomed(\sigma_\exists)$, there is a prefix $h$ of $w$ such that $h \cdot \allpaths \cap \Out(\sigma_\exists) \cap \spec = \varnothing$. 
  Consider the state~$s$ reached of the product reached after reading $h$.
  It is such that no path will be accepting, therefore $s \in D$, and $w$ is rejected by the safety automaton.
  Similarly if a word is rejected by the automaton, the corresponding path in the product reaches a state in $D$, which means it is doomed for $\sigma_\exists$.
\end{proof}

Thanks to this lemma, we can easily exclude doomed histories from the previous construction and recognise $\EA^-(\sigma_\exists)$ which is an ensurable optimal assumption.
\begin{lemma}\label{lem:Moore-ensurable-optimal}
  Given a specification as a parity automaton, and a strategy $\sigma_\exists$ as a Moore machine, we can compute in polynomial time a parity automaton recognising $\EA^-(\sigma_\exists)$.
\end{lemma}
\begin{proof}
  We can  construct a parity automaton which recognises $\EA(\sigma_\exists)$ by Lem.~\ref{lem:auto-EA}.
  By Lem.~\ref{lem:auto-Doomed}, we can construct a safety automaton which recognises $\allpaths \setminus \Doomed(\sigma_\exists)$.
  We consider the product of these two automata and set the colours to be the same than for $\EA(\sigma_\exists)$ except if we reach $\bot$ in the second component, in which case the colour is $1$.
  A word is accepted if, and only if, it belongs to the intersection of the two languages and therefore to $\EA^-(\sigma_\exists)$.
\end{proof}

\begin{theorem}\label{thm:Moore-ensurable-optimal}
  Given a specification as a parity automaton, we can compute in exponential time a parity automaton of polynomial size recognising an ensurable optimal assumption.
  Moreover, if we have access to an oracle for computing memoryless winning strategies in parity games, our algorithm works in polynomial time.
\end{theorem}
\begin{proof}
  We first need to obtain a Moore machine for a memoryless winning strategy $\sigma_\exists$, this can be done in exponential time or constant time if we have an oracle for that.
  Then by Lem.~\ref{lem:Moore-non-dominated}, we can compute a Moore machine implementing a memoryless non-dominated strategy $\sigma'_\exists$.
  By Lem.~\ref{lem:Moore-ensurable-optimal}, we can construct a parity automaton recognising $\EA^-(\sigma'_\exists)$.
  By Lem.~\ref{thm:non-dominated-nonrestrictive}, the language of this automaton is an ensurable optimal assumption.
\end{proof}



\subsection{Scenarios}
\begin{theorem}\label{thm:ea-minus-scenario}
  Let $\spec$ be a specification, and $S$ a coherent scenario compatible with $\spec$.
  If $\sigma_\exists$ is a non-subgame-dominated, then $\EA^-([S \rightarrow \sigma_\exists])$ is $\mathcal{E}$-optimal for $\spec$ and $S$.
\end{theorem}
\begin{proof}
  By Thm.~\ref{thm:ea-scenario} and since $\sigma_\exists$ is subgame winning (Lem.~\ref{lem:non-dominated-strongly-winning}), $\EA([S\rightarrow \sigma_\exists])$ is optimal for $\spec$ and $S$.
  By Lem.~\ref{lem:ea-minus-sufficient}, $\EA^-([S \rightarrow \sigma_\exists])$ is ensurable.
  Assume there are $A \subseteq \allpaths$ and $\sigma'_\exists$ such that $\EA^-([S \rightarrow \sigma_\exists])\subseteq A$, $A$ is ensurable and $A$ sufficient for $\sigma'_\exists$.
  Assume towards a contradiction that there is $w \in A \setminus \EA^-([S \rightarrow \sigma_\exists])$.

  Since $w\not\in \EA^-([S \rightarrow \sigma_\exists])$ it is a losing outcome of $[S \rightarrow \sigma_\exists]$ or $w \in \Doomed([S\rightarrow \sigma_\exists])$.
  We show that in both cases there is $i$ such that $w_{\le i}$ is not prefix of a word in $S$ and there is a losing word $w''$ in $\Out_{w_{\le i}}(\sigma'_\exists)$:
  \begin{itemize}
  \item If $w$ is a losing outcome of $[S \rightarrow \sigma_\exists]$, let $\sigma_\forall$ be such that $w = \Out([S \rightarrow \sigma_\exists],\sigma_\forall)$. 
    Let then $w' = \Out(\sigma'_\exists,\sigma_\forall)$.
    Since $w$ losing and in $A$ which is sufficient for $\sigma'_\exists$, $w$ is not an outcome of $\sigma'_\exists$ and we have that $w \ne w'$.
    Let $i$ be the first index such that $[S \rightarrow \sigma_\exists](w_{\le i}) \ne \sigma'_\exists(w_{\le i})$.
    Since both strategies are compatible with $S$, $w_{\le i}$ is not a prefix of a word in $S$.
    Therefore by definition of $[S \rightarrow \sigma_\exists]$, $w = \Out_{w_{\le i}}(\sigma_\exists)$.
    Since $\sigma_\exists$ is non-subgame-dominated and $w$ losing, there exists a losing outcome~$w''$ in $\Out_{w_{\le i}}(\sigma'_\exists)$.
  \item Otherwise there is $w_{\le k}$ such that $w_{\le k} \cdot \allpaths \cap \Out([S\rightarrow \sigma_\exists]) \ne \varnothing$ and $w_{\le k} \cdot \allpaths \cap \Out([S\rightarrow \sigma_\exists]) \cap \spec = \varnothing$.
    Assume towards a contradiction that $w_{\le k}$ is prefix of an outcome in $S$.
    Let $w'$ be that outcome, we have that $w' \in \Out([S\rightarrow \sigma_\exists]) \cap \spec$ because it is compatible with $S$ which is compatible with $\spec$.
    This is a contradiction with the fact that $w_{\le k} \cdot \allpaths \cap \Out([S\rightarrow \sigma_\exists]) \cap \spec = \varnothing$.
    
    Now, since $w_{\le k}$ is not prefix of an outcome in $S$, $\Out_{w_{\le k}}([S\rightarrow \sigma_\exists]) = \Out_{w_{\le k}}(\sigma_\exists)$.
    Since $\sigma_\exists$ is non-subgame-dominated and $\Out_{w_{\le k}}(\sigma_\exists) \cap \spec = \varnothing$, there is no winning outcome from $w_{\le k}$ (otherwise a strategy allowing this winning outcome would dominated $\sigma_\exists$ from $w_{\le k}$).
    Therefore any word in $w_{\le k} \cdot \allpaths \cap A$ is losing, and since $A$ is sufficient for $\sigma'_\exists$, $w_{\le k} \cdot \allpaths \cap A \cap \Out(\sigma'_\exists) = \varnothing$.
    Let $w_{\le i}$ the greatest prefix of a word in $A$ compatible with $\sigma'_\exists$.
    Since $\sigma_\exists$ is non-subgame-dominated and has a losing outcome from $w_{\le i}$, there exists a losing outcome~$w''$ in $\Out_{w_{\le i}}(\sigma'_\exists)$.
  \end{itemize}
  
  The assumption $\EA^-([S\rightarrow \sigma_\exists]) \cup w''$ is ensurable (as soon as there are at least two input symbols), and it is sufficient for $[S \rightarrow \sigma_\exists]$.
  By optimality of $\EA^-([S\rightarrow \sigma_\exists])$, it includes $w''$.
  Since $A$ includes $\EA^-([S\rightarrow \sigma_\exists])$, it also includes $w''$, which contradicts that $A$ is sufficient for $\sigma'_\exists$.
\end{proof}

\begin{lemma}\label{lem:Moore-shift}
  Given a parity automaton for a coherent scenario $S$ and a strategy $\sigma_\exists$ we can compute in polynomial time a Moore machine for $[S\rightarrow \sigma_\exists]$.
\end{lemma}
\begin{proof}
  We consider the Moore machine where the state space is the product of the parity automaton and the Moore machine representing $\sigma_\exists$.
  We write $D$ for the states of the parity automaton for $S$ from which there is no accepting path.
  If $(s_1,s_2)$ is a state of the product then $\delta''((s_1,s_2), i) = (\delta(s_1,i),\delta'(s_2,i))$, where $\delta$ and $\delta'$ are the transition relation for $S$ and $\sigma_\exists$ respectively.
  Moreover if $s_1 \in D$, then $G'((s_1,s_2)) = G(s_2)$, and otherwise we know that in an output state there is only one~$o$ that leads to a state not in $D$, we then set $G'((s_1,s_2)) = o$.
  This Moore machine implements $[S\rightarrow \sigma_\exists]$.
\end{proof}

\begin{theorem}\label{thm:compute-optimal}
  Given a specification~$\spec$ and a scenario~$S$ as parity automata, we can compute in exponential time a parity automaton of polynomial size recognising a ensurable optimal assumption for $\spec$ and $S$.
\end{theorem}
\begin{proof}
  We can compute in exponential time a memoryless strategy in parity game and as seen in Lem.~\ref{lem:Moore-non-dominated}, we can then compute in polynomial time a memoryless non-dominated strategy $\sigma'_\exists$.
  Since it is memoryless and non-dominated, it is in fact non-subgame-dominated (this should be clear from the definition of non-subgame-dominated).
  Then by Lem.~\ref{lem:Moore-shift}, we can compute a Moore machine for $[S \rightarrow \sigma_\exists]$.
  By Thm.~\ref{thm:ea-minus-scenario}, the corresponding assumption $\EA^-([S \rightarrow \sigma_\exists])$ is ensurable optimal for $\spec$ and $S$.
  By Lem~\ref{lem:Moore-ensurable-optimal}, $\EA^-([S \rightarrow \sigma_\exists])$ can be computed in polynomial time.
\end{proof}

\subsection{Generalisation}

We have seen in Lem.~\ref{lem:generalisation} that from a winning strategy for $H \Rightarrow \spec$ and a strongly winning strategy for $\spec$, we could obtain a strategy $\sigma_\exists$ that has both properties.
Furthermore, we can compute $\sigma'_\exists$ that is strongly non-dominated for $\spec$ and define a strategy that is both non-dominated for $\spec$ and winning for $H \Rightarrow \spec$.
We define $\sigma'_\exists[H \rightarrow \sigma_\exists]$ to be the function that maps $h$ to $\sigma_\exists(h)$ if $h \cdot \sigma_\exists(h)$ is a prefix of some $w \in H$ and maps $h$ to $\sigma'_\exists(h)$ otherwise.

\begin{lemma}\label{lem:ensurable-generalisation}
  If $\sigma_\exists$ is winning for $H \Rightarrow \spec$ and strongly winning for $\spec$ and $\sigma'_\exists$ is strongly non-dominated for $\spec$, then $\EA^-(\sigma'_\exists[H \rightarrow \sigma_\exists])$ contains $H$ and is ensurable-optimal for $\spec$.
\end{lemma}
\begin{proof}
  First notice that $\sigma'_\exists[H \rightarrow \sigma_\exists]$ is winning for $H\Rightarrow \spec$, since by its definition any of its outcome which belongs to $H$ is an outcome of $\sigma_\exists$ and thus satisfies $\spec$.
  We therefore have that $\EA^-(\sigma'_\exists[H \rightarrow \sigma_\exists])$ contains $H$.

  We now show that $\sigma'_\exists[H \rightarrow \sigma_\exists]$ is non-dominated for $\spec$.
  Let $\sigma''_\exists$ be a strategy and $\sigma_\forall$ be an environment strategy such that $w = \Out(\sigma''_\exists,\sigma_\forall)$ is in $\spec$.
  Let $w' = \Out(\sigma'_\exists[H \rightarrow \sigma_\exists],\sigma_\forall)$.
  Let $j$ be the largest index such that $w'_{\le j} = w_{\le j}$.
  Notice that since the environment strategy is the same for both outcomes, the first time they differ is on an output of the system, hence $j$ is odd.
  
  If $w_{\le j+1}$ is a prefix of a word in $H$, then we will follow $\sigma_\exists$ which is strongly winning for $\spec$.
  If there is a winning strategy from $w_{\le j}$ then as  $\sigma_\exists$ is strongly winning for $\spec$, $w$ satisfies $\spec$.
  Otherwise there is a strategy $\sigma_\forall'$ compatible with $w_{\le j}$ such that $\Out(\sigma''_\exists,\sigma_\forall')$ is losing and follows a path of $H$ after $w_{\le j+1}$.
  As $\sigma_\exists$ is winning for $H \Rightarrow \spec$, $\Out(\sigma'_\exists[H \rightarrow \sigma_\exists],\sigma_\forall') \in \spec$.
  
  If $w_{\le j +1}$ is not a prefix of a word in $H$, then we play according to $\sigma'_\exists$ which is strongly non-dominated.
  Hence either $\Out(\sigma'_\exists,\sigma_\forall) \in \spec$ and $\Out(\sigma'_\exists[H \rightarrow \sigma_\exists],\sigma_\forall) \in \spec$ or there exists $\sigma'_\forall$ that is compatible with $w_{\le j}$ and such that $\Out(\sigma'_\exists,\sigma_\forall')\in \spec$ and $\Out(\sigma_\exists'',\sigma_\forall)\not\in \spec$, then since it is compatible with $w_{\le j}$, $\Out(\sigma'_\exists[H \rightarrow \sigma_\exists],\sigma_\forall') = \Out(\sigma'_\exists,\sigma_\forall') \in \spec$.
  This shows that $\sigma'_\exists[H \rightarrow \sigma_\exists]$ is non-dominated.

  Then by Thm.~\ref{thm:non-dominated-nonrestrictive}, $\EA^-(\sigma'_\exists[H \rightarrow \sigma_\exists])$ is ensurable-optimal for $\spec$.
\end{proof}

\begin{lemma}\label{lem:Moore-generalisation}
  Given strategies $\sigma'_\exists$, $\sigma_\exists$ as Moore machines and assumption $H$ as a parity automaton, can construct in polynomial time a Moore machine for $\sigma'_\exists[H \rightarrow \sigma_\exists]$.
\end{lemma}
\begin{proof}
  We write $D$ the set of states of the parity automaton from which there is no accepting path; this can be computed in polynomial time.
  We write $\delta'$ and $\delta$ the transition function for the Moore machine of $\sigma'_\exists$ and $\sigma_\exists$ respectively, $G'$ and $G$ their output functions and $\delta''$ for transition of the parity automaton recognising $H$.
  We consider a Moore machines whose state space is the product of the two Moore machines, the parity automaton and a fourth component in $\{1,2\}$ to mark which strategy to follow; we write $(s,t,u,v)$ a state of the product (note that $s$ corresponds to $\sigma'_\exists$ and $t$ to $\sigma_\exists$).
  We now describe the transition function $\delta'''$ of this Moore machine for a state $(s,t,u,v)$ and an input $i$, and the output function $G''$.
  If $v= 1$ we follow $\sigma_\exists(h)$ then $G''(s,t,u,v) = G(t)$ and otherwise $v=2$, we follow $\sigma'_\exists(h)$ and $G''(s,t,u,v) = G'(s)$.
  If $u$ is an input state and $\delta''(u, G(\delta(t,i)))\not\in D$ then the path following $\sigma_\exists$ is prefix of some word in $H$, so $\sigma'_\exists[H \rightarrow \sigma_\exists]$ follows $\sigma_\exists(h)$, which means for us $\delta'''((s,t,u,v),i)=(\delta'(s,i),\delta(t,i),\delta''(u,i),1)$.
  Otherwise we follow $\sigma'_\exists(h)$, which means $\delta'''((s,t,u,v),i)=(\delta'(s,i),\delta(t,i),\delta''(u,i),2)$.
  This implements the strategy $\sigma'_\exists[H \rightarrow \sigma_\exists]$.
\end{proof}

\begin{theorem}\label{thm:generalization}
  There is an exponential algorithm that given $\spec$ and $H$ sufficient for $\spec$ as parity automata, computes a parity automaton whose language $H'$ is such that $H \subseteq H'$ and $H'$ ensurable-optimal for $\spec$.
\end{theorem}
\begin{proof}
  Assume we are given automata $\calA_H$ for $H$, and $\calA_\spec$ for $\spec$.
  We construct $\calA_{H\Rightarrow\spec}$ recognising $\spec \cup \allpaths \setminus H$.
  We first compute a winning strategy $\sigma_\exists$ in $\calA_{H\Rightarrow\spec}$.
  We can compute in exponential time a memoryless strategy $\sigma_\exists'$ which is winning in $\calA_\spec$ from all states from which there is a winning strategy~\cite{GTW02}, it is in fact strongly winning.
  By Lem.~\ref{lem:Moore-generalisation} we can construct a Moore machine for $\sigma'_\exists[H \rightarrow \sigma_\exists]$, and by Lem.~\ref{lem:Moore-ensurable-optimal} we can obtain an automaton for $\EA^-(\sigma'_\exists[H \rightarrow \sigma_\exists])$.
  Thanks to Lem.~\ref{lem:ensurable-generalisation}, $\EA^-(\sigma'_\exists[H \rightarrow \sigma_\exists])$ is ensurable optimal for $\spec$ and contains $H$.
\end{proof}

\section{Input-assumptions}\label{sec:input}

\subsection{Infinity of optimals}

As we show now, in general there can be an infinite number of incomparable assumptions that are sufficient.

\begin{theorem}\label{thm:infinity}
  There is a specification $\spec$ for which there are an infinite number of optimal input assumptions.
\end{theorem}
\begin{proof}
  Consider the B\"uchi automaton in \figurename~\ref{fig:infinity} with
  $\Sigma_I = \{ i_1, i_2\}$ and $\Sigma_O = \{ o_1 ,o_2\}$, and $L$ its language.
  In the automaton, the objective of the controller is to reach $s_3$.
  For that, it can only influence the transition out of state $s_1$.
  Given $n$, we define the language $A_n$ to be
  $S^\omega \setminus \{(\Sigma_I \cdot \Sigma_O )^n \cdot i_1 \cdot \Sigma_O \cdot\allpaths \}$.
  Under assumption $A_n$, the winning strategy~$\sigma_\exists$ consists in playing $o_2$ at step $n$.
  If assumption $A_n$ is respected  the path goes to $s_3$ and is accepting.
  This shows that $A_n$ is sufficient for $L$.
  
  Each assumption $A_n$ is a safety assumption and an input assumption:
  it is safety because all words not in $A_n$ have a prefix in $(\Sigma_I \cdot \Sigma_O)^n \cdot i_1$ which is a bad prefix, and it is an input assumption because if $w \in A_n \land \pi_I(w) = \pi_I(w')$ then $w'_{2 n + 1} = w_{2 n + 1} \ne i_1$ and therefore $w' \in A_n$.

  We now show that each $A_n$ is optimal.
  Let $B$ be a sufficient assumption such that $A_n \subseteq B$ and $\sigma_\exists$ be a strategy such that $\Out(\sigma_\exists) \cap B \subseteq L$.
  Let $w \not\in A_n$ we will prove $w \not\in B$, which shows $B \subseteq A_n$ and optimality of $A_n$.
  By definition of $A_n$, $w_{2 n + 1} = i_1$ 
  Let $w' = \Out(\sigma_\exists,\pi_I(w))$, we have that $\pi_I(w') = \pi_I(w)$ and since $B$ is an input assumption $w' \in B \Leftrightarrow w \in B$.
  Assume towards a contradiction that $w' \in L$.
  Then there is an index $k \ge 1$ such that $w'_{2 k} = o_2$, $w'_{2 k +1} = i_2$ and for all $1 \le j < k$, $w'_{2 j} = o_1$.
  Since $w_{2n+1} = i_1$, $k \ne n$.
  \begin{itemize}
  \item If $k > n$ then let $w'' =\Out(\sigma_\exists, u)$ where $u = \pi_I(w')_{\le n} \cdot i_2 \cdot {i_1}^\omega$.
    We have $w''_{2 n + 1} = i_2$ therefore $w'' \in A_n \subseteq B$, and $w'' \in \Out(\sigma_\exists)$.
    After reading $w''_{\le 2 n} = w'_{\le 2 n}$ we are still in state $s_0$.
    After that we read $i_2$ which brings us in $s_1$ and then the inputs are all $i_1$, therefore $w''$ is losing which contradicts $B \cap \Out(\sigma_\exists) \subseteq L$.
  \item If $k < n$ then let $w'' =\Out(\sigma_\exists, u)$ where $u_{k+1} = i_1$, $u_{n+1} = i_2$ and $u_j = \pi_I(w')_j$ otherwise.
  We have that $w'' \in A_n \subseteq B$, and $w'' \in \Out(\sigma_\exists)$.
  Since $k < n$, $w''_{\le 2 k} = w'_{\le 2 k}$ and $w''_{2 k + 1} = i_1$, therefore $w''$ is losing which contradicts $B \cap \Out(\sigma_\exists) \subseteq L$.
  \end{itemize}
  This shows that $w' \in \Out(\sigma_\exists) \setminus L$ and since $\Out(\sigma_\exists) \cap B \subseteq L$, $w' \not\in B$ and as a consequence $w\not\in B$.
\end{proof}

\subsection{Link with remorsefree strategies}
In this section we show a link between input assumptions and a class of strategies called remorsefree.
\begin{definition}[Remorsefree]
  Given a specification~$\spec$, a strategy~$\sigma_\exists$ is \emph{remorsefree} if for all~$\sigma'_\exists$ and $w \in \Sigma_I^\omega$,
  $\out(\sigma'_\exists,w) \models \spec$ implies $\out(\sigma_\exists,w) \models \spec$. 
  This is the notion used in~\cite{DF11} for dominance.
  A strategy~$\sigma_\exists$ is \emph{remorsefree-admissible} if for all~$\sigma'_\exists$ either
  $\forall w \in \Sigma_I^\omega.\ \out(\sigma'_\exists,w) \models \spec \Rightarrow \out(\sigma_\exists,w) \models \spec$
  or
  $\exists w \in \Sigma_{I}^\omega.\ \out(\sigma_\exists,w) \models \spec \not\Rightarrow \out(\sigma'_\exists,w) \models \spec$.
\end{definition}

\begin{lemma}
  Given a strategy $\sigma_\exists$ of \eve, if $\spec \ne \varnothing$ then the assumption
  \[
  \IA(\sigma_\exists) = \{ w \in \allpaths \mid \pi_I(w) \not\in \pi_I( \out(\sigma_\exists) \setminus \spec) \}
  \]
  is an input-assumption that is sufficient for $\sigma_\exists$.
  Moreover if $A$ is an input-assumption which is sufficient for $\sigma_\exists$ then $A \subseteq \IA(\sigma_\exists)$.
\end{lemma}
\begin{proof}
  First notice that $\IA$ is an input assumption: if $\pi(w) = \pi(w')$ then $w \in \IA(\sigma_\exists) \Leftrightarrow w' \in \IA(\sigma_\exists)$.

  We now show that $\IA(\sigma_\exists)$ is sufficient for $\sigma_\exists$.
  Let $w \in \IA(\sigma_\exists) \cap \Out(\sigma_\exists)$.
  Since $w \in \IA(\sigma_\exists)$, $\pi_I(w) \not\in \pi_I(\Out(\sigma_\exists) \setminus \spec)$.
  Therefore $w \not\in \Out(\sigma_\exists) \setminus \spec$.
  Hence $w\in \spec$.

  Finally, we show that if $A$ is an input-assumption sufficient for $\sigma_\exists$ then $A \subseteq \IA(\sigma_\exists)$.
  Let $A$ be a sufficient input-assumption and $w \in A$.
  Assume towards a contradiction that $w \not\in \IA(\sigma_\exists)$.
  Then $\pi_I(w) \in \pi_I(\Out(\sigma_\exists) \setminus \spec)$.
  Let $w' \in \Out(\sigma_\exists) \setminus \spec$ such that $\pi_I(w') = \pi_I(w)$.
  Since $A$ is sufficient, $w' \not\in A$.
  Since $A$ is an input-assumption and $\pi_I(w') = \pi_I(w)$ we also have $w \not\in A$ which is a contradiction.
\end{proof}

In~\cite{DF11}, Finkbeiner and al. show the following.
\begin{theorem}{\cite{DF11}}
  There is a remorsefree strategy if, and only if, there is a unique minimal assumption for $\spec$.
  Moreover the minimal assumption is the input-assumption sufficient for the remorsefree strategy.
\end{theorem}

\begin{example}\label{ex:rfree}
Consider the game of \figurename~\ref{fig:rfree-/->dom}.
In this game the only remorsefree strategy is to output $o_1$ at the first step.
The corresponding assumption is $A = \Sigma_I \cdot \Sigma_O \cdot (\{ i_2, i_3 \} \cdot \Sigma_O)^\omega$ while the assumption corresponding to the strategy outputting $o_2$ is $\Sigma_I \cdot \Sigma_O \cdot (\{ i_3 \} \cdot \Sigma_O)^\omega$ which is more restrictive.
The assumption $A$ is indeed the unique optimal input-assumption.
\end{example}

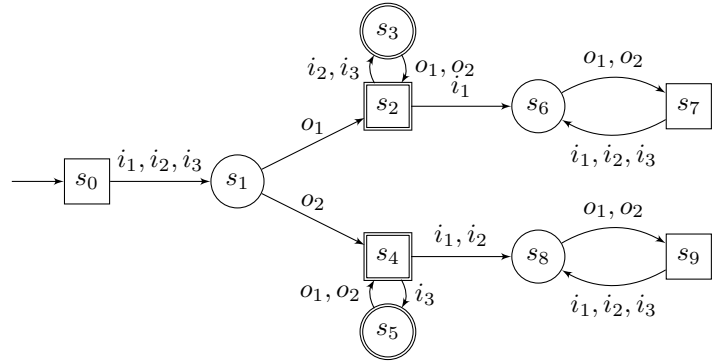
\begin{figure}[h]
  \centering{
  \begin{tikzpicture}
    \draw (-2,0) node[draw,minimum size=6mm] (S0) {$s_0$};
    \draw (0,0) node[draw,circle] (S1) {$s_1$};
    \draw (2,1) node[draw,minimum size=6mm,double] (S2) {$s_2$};
    \draw (2,2) node[draw,double,circle] (S3) {$s_3$};
    \draw (2,-1) node[draw,minimum size=6mm,double] (S4) {$s_4$};
    \draw (2,-2) node[draw,double,circle] (S5) {$s_5$};
    \draw (4,1) node[draw,circle] (L1) {$s_6$};
    \draw (6,1) node[draw,minimum size=6mm] (L2) {$s_7$};
    \draw (4,-1) node[draw,circle] (L3) {$s_8$};
    \draw (6,-1) node[draw,minimum size=6mm] (L4) {$s_9$};

    \draw[-latex'] (-3,0) -- (S0);
    \draw[-latex'] (S0) -- node[above]{$i_1,i_2,i_3$} (S1);
    \draw[-latex'] (S1) -- node[above]{$o_1$} (S2);
    \draw[-latex'] (S1) -- node[above]{$o_2$} (S4);
    \draw[-latex'] (S2) -- node[above] {$i_1$} (L1);
    \draw[-latex'] (S4) -- node[above] {$i_1,i_2$} (L3);
    \draw[-latex'] (L1) edge[bend left] node[above] {$o_1,o_2$} (L2);
    \draw[-latex'] (L3) edge[bend left] node[above] {$o_1,o_2$} (L4);
    \draw[-latex'] (L2) edge[bend left] node[below] {$i_1,i_2,i_3$} (L1);
    \draw[-latex'] (L4) edge[bend left] node[below] {$i_1,i_2,i_3$} (L3);
    \draw[-latex'] (S2) edge[bend left] node[left] {$i_2,i_3$} (S3);
    \draw[-latex'] (S3) edge[bend left] node[right] {$o_1,o_2$} (S2);
    \draw[-latex'] (S5) edge[bend left] node[left] {$o_1,o_2$} (S4);
    \draw[-latex'] (S4) edge[bend left] node[right] {$i_3$} (S5);
  \end{tikzpicture}
  }
  \caption{B\"uchi automaton for which the remorsefree strategy consists in outputting $o_1$.}
  \label{fig:rfree-/->dom}
\end{figure}

We now use the associated notion of admissibility to characterise the minimal assumptions that are sufficient to win.

\begin{theorem}\label{thm:remorse-free-input}
  If $\sigma_\exists$ is a remorsefree admissible strategy for $\spec$, then $\IA(\sigma_\exists)$ is an optimal input-assumption for $\spec$.
  Reciprocally if $A$ is an optimal input-assumption for $\spec$, then there is a remorsefree admissible strategy $\sigma_\exists$, such that $A = \IA(\sigma_\exists)$.
\end{theorem}
\begin{proof}
  \fbox{$\Rightarrow$}
  Let $\sigma_\exists$ be a remorsefree-admissible strategy and $A$ the corresponding environment assumption.
  Let $B$ be such that $A \subset B$.
  We show that $B$ is not sufficient for $\spec$ which will show that $A$ is optimal.
  
  Let $\sigma'_\exists$ be a strategy we show that $B$ is not sufficient for this strategy.
  Since $\sigma_\exists$ is remorsefree-admissible, one of those two cases occurs:
  \begin{itemize}
  \item  $\forall w' \in \Sigma_I^\omega.\ \out(\sigma'_\exists,w') \models \spec \Rightarrow \out(\sigma_\exists,w') \models \spec$.
    Let $w\in B\setminus A$.
    Since $A=\IA(\sigma_\exists)$, we have that $w \in \pi_I(\Out(\sigma_\exists)\setminus \spec)$.
    Hence $\Out(\sigma_\exists,w) \not\models \spec$, and we have that $\out(\sigma'_\exists,w) \not\models \spec$ which shows that $B$ is not sufficient for~$\sigma'_\exists$.
  \item
  or
  $\exists w' \in \Sigma_{I}^\omega.\ \out(\sigma_\exists,w') \models \spec \land \out(\sigma'_\exists,w') \not\models \spec$.
  We have that $w'$ belongs $A$ by definition, thus it belongs to $B$ by hypothesis, and since $\out(\sigma'_\exists,w') \not\models \spec$, $B$ is not sufficient for $\sigma'_\exists$.
  \end{itemize}

  \fbox{$\Leftarrow$}
  Let $A$ be an optimal assumption for $\spec$ and $\sigma_\exists$ the corresponding strategy.
  We will show that $\sigma_\exists$ is remorsefree-admissible.
  For that, let $\sigma'_\exists$ be another strategy of \eve, we show that either:
  $\forall w \in \Sigma_I^\omega.\ \out(\sigma'_\exists,w) \models \spec \Rightarrow \out(\sigma_\exists,w) \models \spec$
  or $\exists w \in \Sigma_{I}^\omega.\ \out(\sigma_\exists,w) \models \spec \not\Rightarrow \out(\sigma'_\exists,w) \models \spec$.

  Since $A$ is optimal, this means that it is not strictly included in $B = \IA(\sigma'_\exists)$.
  This means that either $A = B$ or $A \setminus B \ne \varnothing$.
  \begin{itemize}
  \item if $A=B$ then we show that $\forall w \in \Sigma_I^\omega.\ \out(\sigma'_\exists,w) \models \spec \Rightarrow \out(\sigma_\exists,w) \models \spec$:
    \begin{itemize}
    \item
      if $w \in A = \IA(\sigma_\exists)$, we have that $\Out(\sigma_\exists,w) \models \spec$, and the implication holds.
    \item if $w\not\in A = B = \IA(\sigma'_\exists)$, we have that $\Out(\sigma'_\exists,w) \not\models \spec$, and the implication holds. 
    \end{itemize}
\item
  Otherwise $A \setminus B \ne \varnothing$ then let $w\in A\setminus B$.
  Since $w \in A$, $\Out(\sigma_\exists,w) \models \spec$ and since $w \not\in B = \IA(\sigma'_\exists)$, $\Out(\sigma'_\exists,w) \not\models \spec$.
  This shows that $\out(\sigma_\exists,w) \models \spec \not\Rightarrow \out(\sigma'_\exists,w) \models \spec$.
  \end{itemize}
\end{proof}

\bibliographystyle{abbrv}
\bibliography{biblio}

\begin{thebibliography}{10}

\bibitem{berwanger07}
D.~Berwanger.
\newblock Admissibility in infinite games.
\newblock In {\em Proc. of {STACS}'07}, volume 4393 of {\em LNCS}, pages
  188--199. Springer, Feb. 2007.

\bibitem{BEJK14}
R.~Bloem, R.~Ehlers, S.~Jacobs, and R.~K{\"{o}}nighofer.
\newblock How to handle assumptions in synthesis.
\newblock In {\em Proceedings 3rd Workshop on Synthesis, {SYNT} 2014, Vienna,
  Austria, July 23-24, 2014.}, pages 34--50, 2014.

\bibitem{BRS14}
R.~Brenguier, J.~Raskin, and M.~Sassolas.
\newblock The complexity of admissibility in omega-regular games.
\newblock In {\em {CSL-LICS} '14, 2014}. ACM, 2014.

\bibitem{BRS15}
R.~Brenguier, J.-F. Raskin, and O.~Sankur.
\newblock Assume-admissible synthesis.
\newblock In L.~Aceto and D.~de~Frutos-Escrig, editors, {\em CONCUR}, volume~42
  of {\em LIPIcs}, pages 100--113. Schloss Dagstuhl - Leibniz-Zentrum fuer
  Informatik, 2015.

\bibitem{CHJ08}
K.~Chatterjee, T.~A. Henzinger, and B.~Jobstmann.
\newblock Environment assumptions for synthesis.
\newblock In {\em CONCUR 2008-Concurrency Theory}, pages 147--161. Springer,
  2008.

\bibitem{DF11}
W.~Damm and B.~Finkbeiner.
\newblock Does it pay to extend the perimeter of a world model?
\newblock In {\em FM 2011: Formal Methods}, pages 12--26. Springer, 2011.

\bibitem{EJ91}
E.~A. Emerson and C.~S. Jutla.
\newblock Tree automata, mu-calculus and determinacy.
\newblock In {\em Foundations of Computer Science, 1991. Proceedings., 32nd
  Annual Symposium on}, pages 368--377. IEEE, 1991.

\bibitem{Faella09}
M.~Faella.
\newblock Admissible strategies in infinite games over graphs.
\newblock In {\em {MFCS} 2009}, volume 5734 of {\em Lecture Notes in Computer
  Science}, pages 307--318. Springer, 2009.

\bibitem{FK15}
R.~Faran and O.~Kupferman.
\newblock Spanning the spectrum from safety to liveness.
\newblock In {\em Automated Technology for Verification and Analysis}, pages
  183--200. Springer, 2015.

\bibitem{FJR09}
E.~Filiot, N.~Jin, and J.-F. Raskin.
\newblock An antichain algorithm for {LTL} realizability.
\newblock In {\em Computer Aided Verification}, pages 263--277. Springer, 2009.

\bibitem{GTW02}
E.~Gr\"{a}del, W.~Thomas, and T.~Wilke, editors.
\newblock {\em Automata Logics, and Infinite Games: A Guide to Current
  Research}.
\newblock Springer-Verlag New York, Inc., New York, NY, USA, 2002.

\bibitem{PR89}
A.~Pnueli and R.~Rosner.
\newblock On the synthesis of a reactive module.
\newblock In {\em Proceedings of the 16th ACM SIGPLAN-SIGACT symposium on
  Principles of programming languages}, pages 179--190. ACM, 1989.

\end{thebibliography}



\end{document}